\documentclass{fundam}

\usepackage[utf8]{inputenc}
\usepackage[T1]{fontenc}
\usepackage{amssymb,latexsym}
\usepackage{amsmath}
\usepackage{graphicx}
\usepackage{color}
\usepackage{cleveref}
\usepackage{todonotes}
\usepackage{mathtools}
\usepackage{url}

\def\A{\mathcal A}
\def\B{\mathcal B}

\def\L{\mathcal L}

\def\N{\mathbb N}
\def\uu{\mathbf u}
\def\vv{\mathbf v}

\def\xx{\mathbf x}
\def\yy{\mathbf y}
\def\dd{\mathbf d}

\newcommand{\pext}[2]{\#\mbox{Pext}_{#1}\left(#2\right)}
\newcommand{\CR}[2]{\mathcal{CR}_{#1}\left(#2\right)}

\newcommand{\RR}[2]{\mathcal{R}_{#1}\left(#2\right)}

\newtheorem{thm}{Theorem}
\newtheorem{coro}{Corollary}
\newtheorem{lem}{Lemma}

\newtheorem{prop}{Proposition}
\newtheorem{defi}{Definition}

\usepackage[pagewise]{lineno}

 \usepackage{hyperref}

\begin{document}
	
\setcounter{page}{1}
\publyear{22}
\papernumber{2102}
\volume{185}
\issue{1}

  \finalVersionForARXIV

	\title{On Morphisms Preserving Palindromic Richness}
	
	\author{Francesco Dolce, Edita Pelantov\'a\thanks{Address for correspondence: Czech Technical University in Prague,
                           160 00 Prague 6 - Dejvice,  Czech Republic. \newline \newline
                  \vspace*{-6mm}{\scriptsize{Received March 2021; \ accepted Febryary 2022.}}}
     \\
	Czech Technical University in Prague\\
    160 00 Prague 6 - Dejvice, Czech Republic\\
   {dolcefra@fit.cvut.cz, edita.pelantova@fjfi.cvut.cz}
    }
	
	\maketitle

\runninghead{F. Dolce and E. Pelantov\'a}{On Morphisms Preserving Palindromic Richness}

	\begin{abstract}
		It is known that each word of length $n$ contains at most $n+1$ distinct palindromes.
		A finite \emph{rich word} is a word with maximal number of palindromic factors.
		The definition of palindromic richness can be naturally extended to infinite words.
		Sturmian words and Rote complementary symmetric sequences form two classes of binary rich words, while episturmian words and words coding symmetric $d$-interval exchange transformations give us other examples on larger alphabets.
		In this paper we look for morphisms of the free monoid, which allow us to construct new rich words from already known rich words.
		We focus on morphisms in Class $P_{ret}$.
		This class contains morphisms injective on the alphabet and satisfying a particular palindromicity property: for every morphism $\varphi$ in the class there exists a palindrome $w$ such that $\varphi(a)w$ is a first complete return word to $w$ for each letter $a$.
		We characterize $P_{ret}$ morphisms which preserve richness over a binary alphabet.
		We also study marked $P_{ret}$ morphisms acting on alphabets with more letters.
		In particular we show that every Arnoux-Rauzy morphism is conjugated to a morphism in Class $P_{ret}$ and that it preserves richness.
	\end{abstract}


	\section{Introduction}
	
	Palindromes are words that coincide when read left to right and right to left.
	In natural languages we can find several known, and sometimes amusing, examples of palindromes.
	In the theory of formal languages, questions related to palindromes are usually elegantly formulated but difficult to solve.
	An example of this is the famous HKS conjecture, whose first vague formulation appears in 1995 in~\cite{HoKnSi95}, and that still remains unsolved in its general case.
	According to this conjecture, if the language of a fixed point of a primitive morphism contains infinite many palindromes, then the language  coincides with the language of a fixed point of a morphism from Class $P$ (see Definition~\ref{def:classP} for the definition of Class $P$ and ~\cite{LaPe16} for a more detailed explanation of the HKS conjecture).
	Another recent conjecture is the one about palindromic length of aperiodic words, defined in~\cite{FrPuZa13}, which still resists many attempts to prove it.
	
	This article is devoted to words rich in palindromes.
	The definition of richness is a natural consequence of a simple observation, made by Droubay, Justin and Pirillo in~\cite{DrJuPi01}, that each word of length $n$ contains at most $n+1$ distinct palindromes, counting also the empty word $\varepsilon$ which is considered to be a palindrome of length zero.
	A word of length $n$ is said to be \emph{rich} in palindromes if $n+1$ distinct palindromes occur in it.
	For example, ${\tt ananas}$ (Czech and Italian for "pineapple") is rich in taste but also in palindromes since it contains the 7 palindromes: $\varepsilon, {\tt a}, {\tt n}, {\tt s}, {\tt ana}, {\tt nan}, {\tt anana}$.
	Another rich word is ${\tt pizza}$, since it contains the palindromes $\varepsilon, {\tt a}, {\tt i}, {\tt p}, {\tt z}, {\tt zz}$.
	A counterexample is given by the word ${\tt hawaiianpizza}$ of length $13$ that only contains $12$ palindromes, namely $\varepsilon, {\tt a}, {\tt h}, {\tt i}, {\tt n}, {\tt p}, {\tt w}, {\tt z}, {\tt ii}, {\tt zz}, {\tt awa}, {\tt aiia}$, and thus is not rich in palindromes (nor in taste, according to the first author).
	
	Every factor of a rich word is rich as well.
	Therefore the set of all rich words over a given finite alphabet is a so-called \emph{factorial language}.
	Its factor complexity is known to be subexponential (see ~\cite{Ru17}) and superpolynomial (see~\cite{GuShSh16}).
	Nevertheless the gap between the best known upper and lower bounds on number of rich words of length $n$ over a fixed alphabet is still huge.
	An efficient algorithm to enumerate rich words over a fixed alphabet up to a given length is provided in~\cite{RuSh}.
	To improve the lower bound requires to construct more rich words.
	In this paper we look for rewriting rules (formally morphisms of the free monoid), which allow to construct new rich words from already known rich words.
	An example of such a rewriting rule is the so-called \emph{Fibonacci morphism}: ${\tt 0}$ rewrites to ${\tt 01}$ and ${\tt 1}$ rewrites to ${\tt 0}$.
	Applying the Fibonacci morphism to a rich word, say ${\tt 0001110}$, we get a new rich word, i.e., ${\tt 01010100001}$.
	
	Because of the result of J. Vesti on extendability of rich words (see~\cite{Ve14}), it is enough to apply rewriting rules only to infinite words, i.e., to sequences of the form $\uu = u_0 u_1 u_2 \cdots$, where $u_i$ is an element of a finite alphabet $\A$ for each $i \in\N$.
	The definition of palindromic richness can be naturally extended to infinite words.
	We say that $\uu$ is rich in palindromes if all its finite prefixes are rich in palindromes.
	The most studied class of binary words, namely Sturmian words, consists of words rich in palindromes.
	Therefore rich words can be considered as one of the possible generalizations of Sturmian words.
	The Rote complementary symmetric sequences form another class of binary rich words.
	Except for these two classes, only several singular examples of binary rich words are known (e.g., the period doubling word).
	On $d$-ary alphabets, with $d \geq 3$, two disjoint classes of rich words were found: the episturmian words and the words coding $d$-interval exchange transformation with the symmetric permutation of intervals.
	On ternary alphabets a further class of rich words is described in~\cite{BaPeSt09}.
	
	Here we study morphisms in Class $P_{ret}$.
	A morphism $\varphi$ is in this class if it is injective and if there exists a palindrome $w$ such that $\varphi(a)w$ is a first complete return word to $w$ for each letter $a$.
	This class was introduced in~\cite{BaPeSt11}.
	Harju, Vesti and Zamboni used these morphisms in~\cite{HaVeZa16} to prove a weak version of the HKS conjecture for rich words.
	The palindromic richness of an infinite word can be equivalently described by a property of its bispecial factors, therefore we concentrate on bispecial factors in images of rich words under a morphism in Class $P_{ret}$.
	This enables us to characterize $P_{ret}$ morphisms which preserve richness on a binary alphabet.
	On alphabets with more letters we obtain only partial results for the so-called marked morphisms.
	In particular we show that Arnoux-Rauzy morphisms are conjugated to morphisms in Class $P_{ret}$ and preserve the set of all rich words.
	
	The article is organized as follows.
	In Section~\ref{sec:prel} necessary definitions and tools for the study of richness are presented.
	In Section~\ref{ClassP} we define Class $P_{ret}$ morphisms and list some of their properties, including their relation to morphisms in Class~$P$.
	Marked morphisms in Class $P_{ret}$ are studied in Section~\ref{sec:marked}.
	Section~\ref{sec:main} is devoted to the proofs of our main results on morphisms preserving richness.
	These results are applied in Section~\ref{sec:derived} where we show how to find new rich words.
	Finally, some suggestions for further research are presented in the Section~\ref{sec:conclusions}.

		\section{Preliminaries}	\label{sec:prel}
		
	\subsection{Words}
	
	Let $\A$ be a finite alphabet and $\A^*$ the free monoid over $\A$.
	The elements of $\A$ are called \emph{letters} and the elements of $\A^*$ are called (finite) \emph{words}.
	The \emph{length} of a word $u = a_0 a_1 \cdots a_{n-1}$, with $a_i \in \A$, is the number $|u| = n$.
	The identity of the monoid is called the \emph{empty word} and is denoted by $\varepsilon$.
	Note that $\varepsilon$ is the only word of length $0$.
	If a word $u \in \A^*$ can be written as $u = xwy$, with $x, w, y \in \A^*$, we say that $w$ is a \emph{factor} of $u$, $x$ is a \emph{prefix} of $u$ and $y$ is a \emph{suffix} of $u$.
	If $u = a_0 a_1 \cdots a_n$ and $w = a_i a_{i+1} \cdots a_j$ we say that $i$ is an \emph{occurrence} of $w$ in $u$.
	If $z = xy \in \A^*$, then we denote $x^{-1} z = y$ and $z y^{-1} = x$.
	Given a non-empty word $u = a_0 a_1 \cdots a_{n-1}$, with $a_i \in \A$, we define its \emph{mirror image} $\overline{u}$ as $\overline{u} = a_{n-1} a_{n-2} \cdots a_0$.
	The mirror image of the empty word  is $\overline{\varepsilon} = \varepsilon$.
	A word $u$ is called a \emph{palindrome} (or \emph{palindromic}) if $u = \overline{u}$.
	
	A (right) infinite word is an infinite sequence of letters $\uu = a_0 a_1 a_2 \cdots$.
	We denote by $\A^\N$ the set of infinite words over $\A$.
	In this paper we use bolder letters to denote infinite words.
	We can extend in a natural way to infinite words the notions of factor, prefix, suffix and occurrence.
	An infinite word $\uu$ is called \emph{ultimately periodic} if $\uu = u v^\omega = u v v v \cdots$ for some  $u, v \in \A^*, v \ne \varepsilon$.
	If, moreover, $u$ can be chosen to be empty, the word $\uu$ is said to be (purely) \emph{periodic}.
	A word that is not ultimately periodic is called \emph{aperiodic}.
	A word $\uu$ is called \emph{recurrent} if every factor of $\uu$ appears at least twice (thus an infinite number of times).
	It is \emph{uniformly recurrent} if these factors appear with bounded gaps, i.e., if for every factor $u$ of $\uu$ there exists an integer $n \ge 0$ such that all factors of length $n$ of $\uu$ are of the form $x u y$ for certain $x, y \in \A^*$.
	
\medskip
	The \emph{language} of an infinite word $\uu$, is the set $\L(\uu) \subset \A^*$ of its finite factors, i.e.,
	$$
	\L(\uu) = \left\{ u \in \A^* \; : \; \uu = x u \yy \; \mbox{with } \;  x \in \A^*, \yy \in \A^\N \right\}.
	$$
	We say that the language $\L(\uu)$ is periodic (resp. aperiodic, recurrent, uniformly recurrent) if the word $\uu$ is periodic (resp. aperiodic, recurrent, uniformly recurrent).
	A language $\L(\uu)$ is said to be \emph{closed under reversal} if for every $u \in \L(\uu)$ we have $\overline{u} \in \L(\uu)$ as well.
	It is known that if $\L(\uu)$ is closed under reversal, then $\uu$ is recurrent.
	
	The set of \emph{complete return words} to a word $u \in \L(\uu)$ (with respect to the language $\L(\uu)$) is the set $\CR{\uu}{u}$ of words in $\L(\uu)$ having exactly two factors equal to $u$, one as a proper prefix and the other as a proper suffix.
	It is known that $\uu$ is uniformly recurrent if and only if it is recurrent and for every word $u \in \L(\uu)$ the set $\CR{\uu}{u}$ is finite.
	If $uw$ is a complete return word to $u$, then $w$ is called a (right) \emph{return word} to $u$.
	We denote by $\RR{\uu}{u}$ the set of return words to $u$.
	Clearly $\# \CR{\uu}{u} = \# \RR{\uu}{u}$.
	
\medskip
	Let us consider an infinite word $\uu \in \A^\N$ and its language $\L(\uu)$.
	Given a word $u$ we define the following subsets of $\L(\uu)$:
	$$
	L_{\uu}(u) = \left\{ a \in \A \; : \; au \in \L(\uu) \right\},
	\quad
	R_{\uu}(u) = \left\{ b \in \A \; : \; ub \in \L(\uu) \right\},
	$$
	called respectively the set of left extensions and the set of right extensions of $u$ in $\L(\uu)$, and the subset of $\L(\uu) \times \L(\uu)$
	$$
	B_{\uu}(u) = \left\{ (a,b) \in \A \times \A \; : \; aub \in \L(\uu) \right\}
	$$
	called the set of bi-extensions of $u$ in $\L(\uu)$ (see~\cite{acyclic} for a more detailed explanation).
	
	\begin{example}
		\label{ex:Fibo}
		Let us consider an infinite word $\uu \in \{ {\tt 0, 1} \}^\N$ such that its language $\L(\uu)$ contains the following words of length $2$: ${\tt 00}, {\tt 01}, {\tt 10}$ (but not ${\tt 11}$); and the following words of length $3$: ${\tt 001}, {\tt 010}, {\tt 100}, {\tt 101}$ (but not ${\tt 000}$).
		Then one has
		$$
		\begin{array}{ll}
			L_{\bf u}(\varepsilon) = R_{\bf u}(\varepsilon) = \left\{ {\tt 0,1} \right\},
			&
			B_{\bf u}(\varepsilon) = \left\{ ({\tt 0}, {\tt 0}), ({\tt 0}, {\tt 1}), ({\tt 1}, {\tt 0}) \right\}, \\
			L_{\bf u}({\tt 0}) = R_{\bf u}({\tt 0}) = \left\{ {\tt 0,1} \right\},
			&
			B_{\bf u}({\tt 0}) = \left\{ ({\tt 0}, {\tt 1}), ({\tt 1}, {\tt 0}), ({\tt 1}, {\tt 1}) \right\}, \\
			L_{\bf u}({\tt 1}) = R_{\bf u}({\tt 1}) = \left\{ {\tt 0} \right\},
			&
			B_{\bf u}({\tt 1}) = \left\{ ({\tt 0}, {\tt 0}) \right\}.
		\end{array}
		$$
	\end{example}
	
	A word $u \in \L(\uu)$ is \emph{left special} (resp. \emph{right special}) if $\# L_{\bf u}(u) \ge 2$ (resp $\# R_{\bf u}(u) \ge 2$).
	A \emph{bispecial word} is a word that is both left special and right special.
	We denote by $\B(\uu)$ the set of all bispecial factors of $\uu$.
	We also define the \emph{bilateral order} of a word $u$ as
	$$
	b_{\uu}(u) = \# B_{\uu}(u) - \# L_{\uu}(u) -\# R_{\uu}(u) +1.
	$$
	
	\begin{example}
		Let $\uu$ be the infinite word defined in Example~\ref{ex:Fibo}.
		Both the empty word $\varepsilon$ and the factor ${\tt 0}$ are in $\B(\uu)$, while the factor ${\tt 1}$ is neither left special nor right special.
		Moreover, one can check that $b_{\uu}(\varepsilon) = b_{\uu}({\tt 0}) = b_{\uu}({\tt 1}) = 0$.
	\end{example}
	
	An infinite word over an alphabet $\A$ is \emph{episturmian} if its language is closed under reversal and contains for each $n$ at most one word of length $n$ which is right special.
	An episturmian word is \emph{strict episturmian}, or \emph{Arnoux-Rauzy}, if it has exactly one right special word for each length and moreover each right special factor $u$ is such that $\# R_{\bf u}(u) = \# \A$.
	Note that when $\# \A = 2$ every episturmian word is also a strict episturmian word.
	Episturmian words over a binary alphabet are called \emph{Sturmian} words.

	\subsection{Richness}
	
	In~\cite{DrJuPi01} it is proved that a finite word $w$ has at most $|w| + 1$ factors that are palindromes.
	The \emph{defect} $D(w)$ of a finite word $w$ is defined as the difference between $|w| + 1$ and the actual number of palindromes contained in $w$.
	Finite words with zero defects are called \emph{rich words}.
	If $u$ is a factor of $w$, then $D(u) \le D(w)$.
	In particular, factors of a rich word are rich as well.
	Droubay, Justin and Pirillo demonstrated that a finite word $u$ is rich if and only if every prefix $v$ of $u$ has the so-called Property Ju: the longest palindromic suffix of $v$ is unioccurrent in $v$.
	A direct consequence of Property Ju is the following lemma (see~\cite[Lemma 4.4]{BaPeSt11} for a proof).
	
	\begin{lemma}[\cite{BaPeSt11}]
		\label{alternate}
		Let $u, q \in \A^*$ with $q \ne \overline{q}$.
		Assume that $u$ is rich.
		\begin{enumerate}
			\item If both $q$ and $\overline{q}$ occur in $u$, then occurrences of $q$ and $\overline{q}$ in $u$ alternate, i.e., every complete return word to $q$ in $u$ contains $\overline{q}$ and every complete return word to $\overline{q}$ in $u$ contains $q$.
			\item If $p$ is a factor of $u$ such that $q$ occurs in $p$ exactly once as a prefix of $p$ and $\overline{q}$ occurs in $p$ exactly once as a suffix of $p$, then $p$ is a palindrome.
		\end{enumerate}
	\end{lemma}
	
	The following result gives us a more precise characterization of rich words in case of a binary alphabet.
	
	\begin{lemma}[\cite{PeSt17}]
		\label{Montreal}
		Let $u$ be a finite word over $\{ {\tt 0, 1} \}$.
		The word $u$ is not rich if and only if there exists a non-palindromic word $q$ such that ${\tt 0} q {\tt 0}, {\tt 1} q {\tt 1}, {\tt 0} \overline{q} {\tt 1}$ and ${\tt 1} \overline{q} {\tt 0}$ occur in $u$.
	\end{lemma}
	
	The definition of richness can be extended to infinite words.
	The \emph{defect} of an infinite word $\uu$ is defined as the number (finite or infinite)
	$$
	D(\uu) = \sup \left\{ D(w) \; : \; w \mbox{ is a prefix of } \uu \right\}.
	$$
	An infinite word $\uu = a_0 a_1 a_2 \cdots$ is \emph{rich} if for every $n \in \N$ the prefix $a_0 a_1 \cdots a_{n-1}$ of length $n$ contains exactly $n + 1$ different palindromes (see~\cite{DrJuPi01}).
	Infinite words with zero defect correspond exactly to rich words.
	Infinite words with finite (but not necessarily zero) defect are called \emph{almost rich}.
	Such words have been studied in~\cite{GlJuWiZa09}.
	A simple test for richness in periodic words was provided by Brlek, Hamel, Nivat and Reutenauer in~\cite[Theorems 4 and 6]{BrHaNiRe04}.
	
	\begin{thm}[\cite{BrHaNiRe04}]
		\label{periodicRich}
		If $\uu$ is periodic and $\L(\uu)$ contains infinitely many palindromes, then there exist two palindromes $p, q$ such that $\uu = (pq)^{\omega}$.
		
		Moreover, $\uu$ is rich if its prefix of length
		$|pq|+ \max \left\{ \left\lfloor \tfrac{|q|-|p|}{3} \right\rfloor \,, \left\lfloor \tfrac{|p|-|q|}{3} \right\rfloor \right\}$
		is rich.
	\end{thm}
	
	To test richness in general infinite words we introduce the following notion.
	Given an infinite word $\uu \in \A^\N$ and a palindrome $p \in \L(\uu)$, we define the number of palindromic extensions of $p$ in $\uu$ as
	$$
	\pext{\uu}{p} =
	\# \left\{ a \in \A : apa \in \L(\uu) \right\}.
	$$
	
	The following relation between bispecial factors and richness in the language of an infinite word has been proved in~\cite[Theorem 11]{BaPeSt10}.
	
	\begin{thm}[\cite{BaPeSt10}]
		\label{OnlyBSfactors}
		Let $\uu \in \A^\N$ be such that $\L(\uu)$ is closed under reversal.
		Then $\uu$ is rich if and only if every bispecial factor $u \in \L(\uu)$ satisfies
		\begin{equation}
			\label{condition}
			b_{\bf u}(u) =
			\left\{
			\begin{array}{cl}
				\pext{\uu}{u} - 1 & \quad \mbox{if } u \mbox{ is a palindrome,} \\
				0	& \text{\quad otherwise.}
			\end{array}
			\right.
		\end{equation}
	\end{thm}
	
	Extendability of finite rich words into infinite rich words were studied by Vesti. 
	He proved the following property (see~\cite[Propositions 2.11 and 2.12]{Ve14}).
	
	\begin{prop}[\cite{Ve14}]
		\label{ExtensionToPeriodic}
		Let $u$ be a finite rich word.
		There exist an infinite aperiodic rich word and an infinite periodic rich word such that $u$ is a factor of both of them.
	\end{prop}

	\subsection{Morphisms}
	
	Let $\A, \B$ be two alphabets.
	A \emph{morphism} $\varphi: \A^* \to \B^*$ is a monoid morphism from $\A^*$ to $\B^*$.
	A morphism $\varphi: \A^* \to \B^*$ is called \emph{non-erasing} if $|\varphi(a)| > 0$ for all letters $a \in \A$.
	
	Let us consider a morphism $\varphi: \A^* \to \A^*$.
	If $a \in \A$ is such that $\varphi(a)$ begins with $a$ and $|\varphi^n(a)|$ tends to infinity with $n$, then the morphism is called a \emph{substitution}.
	The unique word which has all words $\varphi^n(a)$ as prefixes
	is called a \emph{fixed point} of the substitution $\varphi$ and is denoted by $\varphi^\omega(a)$.
	
	A morphism $\varphi: \A^* \to \A^*$ is called \emph{primitive} if there exists an integer $k$ such that for all $a, b \in \A$, the letter $b$ appears as a factor in $\varphi^k(a)$.
	If $\varphi$ is a primitive morphism, then the language $\L(\uu)$ is uniformly recurrent for all fixed points $\uu$ of $\varphi$ (see~\cite[Proposition 13]{PytheasFogg}).
	
	A morphism $\varphi :\A^* \to \A^*$ is \emph{right conjugated} to a morphism $\psi :\A^* \to \A^*$ (or is a \emph{right conjugate} of $\psi$) if there exists a word $x \in \A^*$, called the \emph{conjugate word}, such that $\varphi(a)x = x \psi(a)$ for each $a \in \A$.
	A \emph{left conjugated} morphism is defined symmetrically.
	We say that $\psi$ is the \emph{rightmost conjugate} of $\varphi$, if $\psi$ is right conjugated to $\varphi$ and $\psi$ is the only right conjugated to itself.
	In this case we denote such rightmost conjugate by $\varphi_R$.
	Analogously we define the \emph{leftmost conjugate} of $\varphi$ and denote it by $\varphi_L$.
	
	If there exists no rightmost conjugate (resp. no leftmost conjugate) of $\varphi$, then $\varphi$ is conjugated to itself by a non-empty conjugate word $x$.
	Such a morphism is called \emph{cyclic} and there exists $z \in \A$ such that $\varphi(a) \in z^*$ for each letter $a \in \A$.
	A fixed point of a cyclic morphism if periodic (see~\cite{Lothaire}).
	A morphism that is not cyclic is called \emph{acyclic}.
	Note that an acyclic morphism can have a periodic fixed point (see, e.g., the morphism $\varphi_2$ in Example~\ref{ex:psi1} later).
	
\medskip
	Over an alphabet $\A$ we define the morphisms $\psi_a$ and $\overline{\psi}_a$ for each $a \in \A$ as follows:
	\begin{equation}
		\label{A-R_Morph}
		\psi_a:
		\left\{
		\begin{array}{lll}
			a \mapsto a & \\
			b \mapsto ab &\text{if} & b \ne a
		\end{array}
		\right.
		\quad \text{and} \qquad
		\overline{\psi}_a:
		\left\{
		\begin{array}{lll}
			a \mapsto a & &\\
			b \mapsto ba & \text{if}  & b\ne a
		\end{array}
		\right.
	\end{equation}
	These morphisms together with permutations of letters in $\A$ are called \emph{elementary Arnoux-Rauzy morphisms}.
	All elementary Arnoux-Rauzy morphisms generate a monoid of morphisms called the \emph{Arnoux-Rauzy monoid}.
	A morphism from this monoid is said to be an \emph{Arnoux-Rauzy morphism}.
	Arnoux-Rauzy morphisms are known to send episturmian words to episturmian words (see~\cite{JuPi02}).
	
	\begin{example}
		\label{ex:varphi}
		Let us consider the morphism $\phi : \{ {\tt 0, 1} \}^* \to \{ {\tt 0, 1} \}^*$ defined by
		$$
		\phi:
		\left\{
		\begin{array}{l}
			{\tt 0} \mapsto {\tt 01} \\
			{\tt 1} \mapsto {\tt 0}
		\end{array}
		\right.
		$$
		Such a morphism, called the \emph{Fibonacci morphism}, is an Arnoux-Rauzy morphism since it can be obtained as the composition $\phi = \psi_{\tt 0} \circ \pi$, where $\pi = ({\tt 1, 0})$ is the permutation interchanging the two letters.
		Its fixed point ${\bf f} = \varphi^\omega({\tt 0})$ is called the \emph{Fibonacci word}.
		
\medskip
		Similarly, the morphism $\tau: \{ {\tt 0, 1, 2} \}^* \to \{ {\tt 0, 1, 2} \}^*$ defined by
		$$
		\tau:
		\left\{
		\begin{array}{l}
			{\tt 0} \mapsto {\tt 01} \\
			{\tt 1} \mapsto {\tt 02} \\
			{\tt 2} \mapsto {\tt 0}
		\end{array}
		\right.
		$$
		is an Arnoux-Rauzy morphism, called the \emph{Tribonacci morphism}.
		It can be obtained as the composition $\tau = \psi_{\tt 0} \circ \pi'$, where $\pi' = ({\tt 1, 2, 0})$.
		Its fixed point ${\bf t} = \tau^\omega({\tt 0})$ is called the \emph{Tribonacci word}.
	\end{example}
	
	Fixed points of Arnoux-Rauzy morphisms are known to be rich.
	
	A morphism $\psi$ over the binary alphabet $\A = \{ {\tt 0,1} \}$ is called a \emph{standard Sturmian morphism} if it belongs to the monoid generated by $\pi$ and $\phi$, where $\pi$ is the permutation of the two letters and $\phi$ is the Fibonacci morphism defined in Example~\ref{ex:varphi}.

 \section{Class $P_{ret}$}
	\label{ClassP}
	
	The basic set in which we will look for morphisms preserving richness is the so-called Class $P_{ret}$.
	This class was introduced in~\cite{BaPeSt11} to study relations between rich and almost rich words.
	
	\begin{defi}
		\label{classPret}
		A morphism $\varphi: \A^* \mapsto \B^*$ belongs to Class $P_{ret}$, if there exists a palindrome $w$ such that
		\begin{itemize}
			\item $\varphi(a) w$ is a palindromic complete return word to $w$ for each $a \in \A$,
			\item $\varphi(a) \ne \varphi(b)$ for each pair $a,b \in \A$, with $a \ne b$.
		\end{itemize}
	\end{defi}
	
	Clearly, if $\varphi \in P_{ret}$, then $\varphi$ is non-erasing.
	Indeed a complete return word to $w$ is longer than $w$ itself.
	In this paper we focus on $P_{ret}$ morphisms for which $\A = \B$.
	In Proposition~\ref{pro:unicity} we will prove that the palindrome $w$ in Definition~\ref{classPret} is unique.
	Following~\cite{HaVeZa16}, we call it the \emph{marker} of the morphism $\varphi \in P_{ret}$.
	
	\begin{example}
		\label{ex:psi1}
		One can easily check that the following standard Sturmian morphisms
		$$
		\varphi_1:
		\left\{
		\begin{array}{l}
			{\tt 0} \mapsto {\tt 010} \\
			{\tt 1} \mapsto {\tt 01001}
		\end{array}
		\right.
		\quad \mbox{and} \quad
		\varphi_2:
		\left\{
		\begin{array}{l}
			{\tt 0} \mapsto {\tt 1110} \\
			{\tt 1} \mapsto {\tt 1}
		\end{array}
		\right.
		$$
		belong to Class $ P_{ret}$ with markers respectively $w_1 = {\tt 010010}$ and $w_2 = {\tt 111}$.
	\end{example}
	
	\begin{example}
		\label{ex:sigma}
		Let $p,q, \ell$ be positive integers, with $p \ne q$.
		The morphism
		$$
		\sigma:
		\left\{
		\begin{array}{l}
			{\tt 0} \mapsto {\tt 0}^\ell {\tt 1}^p \\
			{\tt 1} \mapsto {\tt 0}^\ell {\tt 1}^q
		\end{array}
		\right.
		$$
		belongs to Class $P_{ret}$ and the associated marker is $w = {\tt 0}^\ell$.
	\end{example}
	
	\begin{example}
		\label{Ex:AR}
		For each $a \in \A$, the elementary Arnoux-Rauzy morphism $\psi_a$ defined in~(2) belongs to Class $ P_{ret}$ and $w = a$ is its marker.
		A morphism which permutes letters of the alphabet, i.e., a permutation on $\A$, is in Class $P_{ret}$ as well, the associated marker being $w = \varepsilon$.
		The morphism $\overline{\psi}_a$ does not belong to Class $P_{ret}$, but is conjugated to $\psi_a \in P_{ret}$ with conjugate word $a$.
	\end{example}
	
	\begin{coro}
		Each Arnoux-Rauzy morphism is conjugated to a morphism in Class $P_{ret}$.
	\end{coro}
	\begin{proof}
		As proven in~\cite[Lemma 2.5]{HaVeZa16}, if a morphism $\varphi$ is conjugated to $\varphi'$ and $\psi$ is conjugated to $\psi'$, then their composition $\varphi \circ \psi$ is conjugated to $\varphi' \circ \psi'$.
		Using the previous example, every elementary Arnoux-Rauzy morphism is conjugated to a morphism in Class $P_{ret}$.
		As shown in~\cite{BaPeSt11}, the Class $P_{ret}$ is closed under composition and thus every composition of elementary Arnoux-Rauzy morphisms belongs to Class $P_{ret}$.
	\end{proof}
	
	Let us list some properties of morphisms in Class $P_{ret}$.
	They are stated in~\cite{BaPeSt11} as Remark 5.2 and Proposition 5.3.
	
	\begin{lem}[\cite{BaPeSt11}]
		\label{palindromPreserved}
		Let $\varphi, \sigma: \A^* \mapsto \A^*$ be morphisms in Class $P_{ret}$.
		Denote by $w_\varphi$ and $w_\sigma$ the markers associated to $\varphi$ and $\sigma$, respectively.
		For every $u, v \in \A^*$:
		\begin{enumerate}
\itemsep=0.95pt
			\item $\varphi(u) = \varphi(v)$ implies $u = v$,
			\item $\overline{ \varphi(u) w_\varphi } = \varphi( {\overline{u}} ) w_\varphi$,
			\item $u$ is a palindrome if and only if $\varphi(u) w_\varphi$ is a palindrome,
			\item $\sigma \circ \varphi$ is a morphism in Class $P_{ret}$ and its associated marker is $\sigma(w_\varphi) w_\sigma$.
		\end{enumerate}
	\end{lem}
	
	\begin{example}
		\label{ex:psi2}
		Let us consider the morphism $\varphi_1$ and $\varphi_2$ of Example~\ref{ex:psi1}.\\
		One has $\overline{\varphi_2({\tt 01}) w_2} = \overline{\tt 1110 1 111} = {\tt 111 1 0111} = \varphi_2({\tt 10}) {\tt 111}$.
		Moreover one has that both ${\tt 010}$ and $\varphi_2({\tt 010}) w_2 = {\tt 1110 1 1110 111}$ are palindromes.
		Finally one can check that $\varphi_1 \circ \varphi_2$ is in Class $P_{ret}$ with associated marker $\varphi_1({\tt 111}) {\tt 010010}$.
	\end{example}
	
	Using Lemma~\ref{palindromPreserved}, we deduce several auxiliary useful properties.
	
	\begin{lem}
		\label{lem:wprefix}
		Let $\varphi: \A^* \mapsto \A^*$ be a morphism in Class $P_{ret}$ with marker $w$ and let $\uu \in \A^\N$.
		\begin{enumerate}
   \itemsep=0.95pt
			\item If $v \in \A^*$, then $w$ occurs in $\varphi(v) w$ exactly $|v| + 1$ times and $w$ is a prefix of $\varphi(v)w$.
			\item The set of return words to $w$ in $\varphi(\uu)$ is $\RR{\varphi(\uu)}{w} = \{ \varphi(a) : a \in \A \cap \L(\uu) \}$.
			\item If $u \in \L(\uu) $, then $\varphi(u) w \in \L(\varphi(\uu))$.
		\end{enumerate}
\end{lem}

\eject
	\begin{proof}

\vspace*{-4mm}
		\textbf{Item (1).}
		We proceed by induction on $|v|$.
		The claim is clearly valid when $v = \varepsilon$.
		Assume now that $v = v'a$, with $a \in \A$.
		Then $\varphi(v) w = \varphi(v') \varphi(a) w$.
		By definition of $P_{ret}$, the factor $w$ is a prefix of $\varphi(a)w$ and it has 2 occurrences in $\varphi(a) w$.
		Therefore we can write $\varphi(v) w = \varphi(v') w w^{-1} \varphi(a) w$.
		By induction hypothesis, $w$ is a prefix of $\varphi(v') w$ and occurs $|v'| + 1 = |v|$ times in $\varphi(v') w$.
		Consequently, $w$ is a prefix of $\varphi(v) w = \varphi(v') w w^{-1} \varphi(a) w$ and has $|v| + 1$ occurrences in $\varphi(v) w$.

\medskip		
		\textbf{Item (2).}
		Let $\uu = a_0 a_1 a_2 \cdots$.
		First we show that for every $i \in \N$, the word $\varphi(a_i a_{i+1} a_{i+2} \cdots)$ has $w$ as a prefix.
		Indeed, the word $\varphi(a_i a_{i+1} \cdots a_{i+k})$ is longer than $|w|$ for some $k \in \N$.
		By Item $(1)$, the word $\varphi(a_i a_{i+1} \cdots a_{i+k})w$ has $w$ as a prefix.
		Therefore, $w$ is a prefix of $\varphi(a_i a_{i+1} \cdots a_{i+k})$ as well.
		
		Since $w$ is a prefix both of $\varphi(a_i a_{i+1} a_{i+2} \cdots)$ and $\varphi(a_{i+1} a_{i+2} a_{i+3} \cdots)$, the word $\varphi(a_i) w$ is a prefix of $\varphi(a_i a_{i+1} a_{i+2} \cdots)$.
		By definition of Class $P_{ret}$, $w$ occurs twice in $\varphi(a_i) w$.
		This implies that $\varphi(a_i)$ is a return word to $w$ in $\varphi(\uu)$.

\medskip		
		\textbf{Item (3).}
		Let $u \in \L(\uu)$.
		Since the language $\L(\uu)$ is extendable and the morphism $\varphi$ is non-erasing, there exist letters $c_1, c_2, \ldots, c_k \in \A$ such that $u  c_1 c_2 \cdots c_k \in \L(\uu)$ and $|\varphi (c_1 c_2 \cdots c_k)| \ge |w|$.
		Thus $\varphi(u  c_1 c_2 \cdots c_k)$ is an element of $\L(\varphi(\uu))$.
		It is enough to show that $w$ is a prefix of $\varphi(c_1 c_2 \cdots c_k)$.
		We know from Item (1) that $w$ is a prefix of $\varphi(c_1 c_2 \cdots c_k)w$.
		Since $|\varphi(c_1 c_2 \cdots c_k)| \ge |w|$, the word $\varphi(c_1 c_2 \cdots c_k)$ starts with $w$ as well.
	\end{proof}
	
	\begin{example}
		Let us consider again the morphism $\varphi_2$ from Examples~\ref{ex:psi1} and~\ref{ex:psi2}.
		According to Lemma~\ref{lem:wprefix}, the word $w_2 = {\tt 111}$ appears exactly three times in $\varphi_2({\tt 01}) {\tt 111}$.
	\end{example}
	
	\begin{prop}
		\label{pro:unicity}
		Each morphism in Class $P_{ret}$ has a unique marker.
	\end{prop}
	\begin{proof}
		Let $\varphi$ be a morphism in Class $P_{ret}$.
		Assume by contradiction that $w$ and $u$ are two distinct markers associated to $\varphi$.
		From Item (1) of Lemma~\ref{lem:wprefix} it follows that both $w$ and $u$ are prefixes of $\varphi(v)$ for a long enough word $v \in \A^*$.
		Therefore $|w| \ne |u|$.
		Without loss of generality, let us assume that $|w| > |u|$.
		In particular, $u$ is a non-empty prefix of $w$, i.e., $w = uz$ for some non-empty word $z \in \A^*$.
		Since both $w$ and $u$ are palindromes, we get $uz = \overline{z} u$.
		It is a well-known fact that this equations on words implies that there exist words $x,y \in \A^*$ and $i \in \N$ such that $z = yx$, $\overline{z} = xy$ and $u = (xy)^i x$.
		Since $w = uz = (xy)^{i+1} x = xyu = uxy$, the marker $u$ has at least two occurrences in $w$.
		Let us discuss two cases.\medskip
		
		1) There exists a letter $a \in \A$ such that $|\varphi(a)| \geq |w|$.
		Then $\varphi(a)$ has a prefix $w$ and thus $\varphi(a)u$ has at least three occurrences of $u$ - a contradiction.\medskip
		
		2) For all letters $a \in \A$ one has $|\varphi(a)| < |w|$.
		Let us fix $a \in \A$ and denote
		$$
		\varphi(a)w = v_0 v_1 \cdots v_{n-1} v_n \cdots v_{n+m-1},
		$$
		where $n = |\varphi(a)|$, $m = |w|$ and $v_i \in \A$ for all $i$.
		Then $0$ and $n$ are the two occurrences of $w$ in $\varphi(a)w$.
		The form of $u$ implies that $0$, $|xy|$, and $n$ are occurrences of $u$ in $\varphi(a)u$.
		As the marker $u$ occurs in $\varphi(a)u$ exactly twice, and $|z| > 0$, we have $|xy| = n$ and hence $\varphi(a) = xy$.
		Since this is true for every letter $a \in \A$, we have a contradiction with the injectivity of $\varphi$.
	\end{proof}
	
	\begin{coro}
		\label{preimage}
		Let $\varphi: \A^* \mapsto \A^*$ be a morphism in Class $P_{ret}$ with marker $w$.
		Let $\uu, \vv \in \A^\N$, with $\uu = a_0 a_1 a_2 \cdots$, $\vv = b_0 b_1 b_2 \cdots$, such that $\vv = \varphi(\uu)$.
		\begin{itemize}
			\item If $v \in \L(\vv)$ has $w$ both as a prefix and as a suffix, then there exists a unique $u \in \A^*$ such that $v = \varphi(u) w$.
			Moreover, if an index $i$ is an occurrence of $v$ in $\vv$, then $b_i b_{i+1} b_{i+2} \cdots = \varphi(a_j a_{j+1} a_{j+2}\cdots)$, where the index $j$ is an occurrence of $u$ in $\uu$.
			\item $\L(\uu)$ is closed under reversal if and only if $\L(\vv)$ is closed under reversal.
		\end{itemize}
	\end{coro}
	\begin{proof}
		Let $v'$ such that $v = v' w$.
		Then $v'$ is a concatenation of return words to $w$.
		By Item (2) of Lemma~\ref{lem:wprefix}, the factor $v$ can be written as $v = \varphi(u)w$ for a certain word $u$.
		Since $\varphi(a) \ne \varphi(b)$ for every $a \ne b$, such a factor $u$ is given uniquely.
		Moreover, Item (2) of Lemma~\ref{lem:wprefix} implies that $w$ occurs in $\vv$ only as a prefix of $\varphi(a_j a_{j+1} a_{j+2}\cdots)$ for some $j$.
		
		To prove the second part of the corollary, we first assume that $\L(\uu)$ is closed under reversal and let $v \in \L(\varphi(\uu))$.
		Then, there exists a factor $u \in \L(\uu)$ such that $v$ occurs in $\varphi(u)$.
		By Item (3) of Lemma~\ref{lem:wprefix}, $\varphi(u)w \in \L(\varphi(\uu))$.
		Since $\overline{u} \in \L(\uu)$, it follows from Item (2) of Lemma~\ref{palindromPreserved} that the factor $\overline{\varphi(u) w} =\varphi(\overline{u})w$ belongs to $ \L(\varphi(\uu))$.
		Obviously, $\overline{v}$ occurs in $\overline{\varphi(u)}$ and thus belongs to $\L(\varphi(\uu))$ too.
		
		Now, let us assume that $\L(\varphi(\uu))$ is closed under reversal and let $u \in \L(\uu)$.
		Then $\varphi(u)w \in \L(\varphi(\uu)) $ and thus, using Item (2) of Lemma~\ref{palindromPreserved}, we have $\overline{\varphi(u) w} = \varphi(\overline{u}) w \in \L(\vv)$.
		By the first statement of this corollary, $\overline{u}$ belongs to $\L(\uu)$.
	\end{proof}
	
	\begin{example}
		\label{ex:varphi2}
		Let us consider the morphism $\phi$ of Example~\ref{ex:varphi} and its fixed point ${\bf f}$.
		Such a morphism is in Class $P_{ret}$ and its marker is $w = {\tt 0}$.
		The word $v = {\tt 010}$ starts and ends with ${\tt 0}$.
		The letter ${\tt 0}$ is the unique word $u \in \mathcal{L}({\bf f})$ such that $v = \phi(u) w$.
	\end{example}
	
	\begin{remark}
		If $\varphi \in P_{ret}$, then every factor of $\varphi(\uu)$ can be decomposed (up to a short prefix and suffix) uniquely into concatenation of return words to $w$.
		Therefore, the synchronizing delay as defined in~\cite{Ca94} is at most the length of a factor which does not contain $w$.
		In other words, the synchronizing delay of $\varphi$ is at most $\max \left\{ |\varphi(a)w| : a \in \A \right\} - 2$.
	\end{remark}
	
	\begin{proposition}
		\label{cyclic}
		Every morphism in Class $P_{ret}$ is acyclic.
	\end{proposition}
	\begin{proof}
		Assume that $\varphi$ is a cyclic morphism in Class $P_{ret}$ with marker $w$.
		Then, there exist $z \in \A^*$, distinct letters $a, b \in \A$ and positive integers $m, n$, with $m < n$ such that $\varphi(a) = z^m$ and $\varphi(b) = z^n$.
		Moreover, $\varphi(a) w = z^m w$ and $\varphi(b) w = z^n w$ are complete return words to $w$.
		In particular, there exists a non-empty word $u \in \A^*$ such that $z^m w = wu$.
		Hence we have $z^n w = z^{n-m} z^m w = z^{n-m} w u$, which implies that $w$ appears at least three times as a factor of $\varphi(b) w$, a contradiction.
	\end{proof}
	
	Let ${\rm Fst}(u)$ and ${\rm Lst}(u)$ denote respectively the first and the last letter of a non-empty word $u$.
	Given a non-erasing morphism $\varphi$ we define the two maps on $\A$
	$$
	\rho_{\varphi}: a \mapsto {\rm Lst} \bigl( \varphi_R(a) \bigr)
	\quad \mbox{and} \quad
	\lambda_{\varphi}: a \mapsto {\rm Fst} \bigl( \varphi_L(a) \bigr)
	$$
	
	\begin{prop}
		\label{conjugacy}
		Let $\varphi$ be a morphism in Class $P_{ret}$ and $w$ be its marker.
		\begin{enumerate}
			\item If $\psi$ is a right conjugate of $\varphi$, then $\psi$ belongs to Class $P_{ret}$.
			In particular, $\varphi_R \in P_{ret}$.
			\item If $\varphi = \varphi_R$, then the marker $w$ is the conjugate word between $\varphi_L$ and $\varphi_R$, i.e.,
			$$
			\varphi_R(a) w = w \varphi_L(a) \quad \text{for all } a \in \A.
			$$
			Moreover, $\rho_\varphi(a) = \lambda_\varphi(a)$.
		\end{enumerate}
	\end{prop}
	\begin{proof}
		First note that, by Proposition~\ref{cyclic}, the morphism $\varphi$ is acyclic and thus $\varphi_R$ and $\varphi_L$ are well defined.
	
\medskip	
		\textbf{Item (1).}
		Let us assume that there exists a letter $x \in \A$ such that $\psi(a) x = x \varphi(a)$ for each $a \in \A$.
		Then $x$ is a suffix of $\varphi(a)$ for each $a \in \A$.
		Since $\varphi(a)w$ is a palindrome, then $wx$ is a prefix of $\varphi(a)w$.
		Thus $\psi(a) x w x = x \varphi(a) w x$ is a palindrome as well and the palindrome $p = xwx$ occurs in $\psi(a) xwx$ exactly twice (otherwise we would have a third occurrence of $w$ as a factor of $\varphi(a) w$).
		Therefore, $\psi \in P_{ret}$ and $p = xwx$ is the associated marker of $\psi$.
		We have shown that $\psi$ obtained from $\varphi$ by right conjugation using one letter $x$ belongs to Class $P_{ret}$.
		As each right conjugation can be obtained step by step by conjugations using one letter only, the assertion is proven.

\medskip		
		\textbf{Item (2).}
		Since $w$ is a prefix of $\varphi(a)w$, we can define a new morphism $\xi$ by setting $\xi(a) = w^{-1} \varphi (a)w$ for each $a \in \A$.
		Obviously, $\xi$ is left conjugated to $\varphi$ and $w$ is the conjugate word.
		Since $\varphi(a)w$ is a palindrome, ${\rm Fst} \left( \xi(a) \right) = {\rm Lst} \left( \varphi(a) \right)$ for each $a \in \A$.
		By our assumption, $\varphi$ is rightmost conjugate and thus there exist two letters, say $b_1$ and $b_2$, such that ${\rm Lst} \left( \varphi(b_1) \right) \ne {\rm Lst} \left( \varphi(b_2) \right)$.
		Therefore, ${\rm Fst} \left( \xi(b_1) \right) \ne {\rm Fst} \left( \xi(b_2) \right)$ and thus $\xi$ is the leftmost conjugate of $\varphi$.
	\end{proof}
	
	\begin{example}
		\label{ex:sigma2}
		Let us consider the morphism $\sigma$ from Example~\ref{ex:sigma} with $\ell = p = 1$ and $q = 3$.
		It can be checked that this morphism is not Sturmian.
		Its rightmost conjugate and leftmost conjugate are, respectively, the morphisms
		$$
		\sigma_R:
		\left\{
		\begin{array}{l}
			{\tt 0} \mapsto {\tt 10} \\
			{\tt 1} \mapsto {\tt 1011}
		\end{array}
		\right.
		\quad \mbox{and} \quad
		\sigma_L:
		\left\{
		\begin{array}{l}
			{\tt 0} \mapsto {\tt 01} \\
			{\tt 1} \mapsto {\tt 1101}
		\end{array}.
		\right.
		$$
		The morphisms $\sigma_R$ belongs to Class $P_{ret}$ and its marker is ${\tt 101}$.
		Moreover, for all $a \in \{ {\tt 0, 1} \}$ one has $\sigma_R(a) {\tt 101} = {\tt 101} \sigma_L(a)$ and $\rho_\sigma(a) = \lambda_\sigma(a)$.
	\end{example}
	
	The following class was introduced by Hof, Knill and Simon in~\cite{HoKnSi95} in order to construct words containing infinitely many palindromic factors.
	Such words play an important role in the study of the spectrum of the Schrödinger operator associated with an aperiodic sequence.
	
	\begin{defi}
		\label{def:classP}
		A morphism $\eta : \A^* \to \A^*$ belongs to Class $P$, if there exists a palindrome $p\in \A^*$ such that $\eta(a) = pq_a$ for each $a \in \A$, where $q_a$ is a palindrome.
	\end{defi}
	
	It is easy to see that every fixed point of a substitution in Class $P$ contains infinitely many palindromes.
	Let us explain the relation between Class $P$ and Class $P_{ret}$.
	
	As demonstrated in~\cite[Lemma 20]{LaPe16}, an acyclic morphism $\varphi$ is conjugated to a morphism in Class $P$ if and only if $\varphi_R(a) = \overline{\varphi_L(a)}$ for each $a \in \A$.
	If it is the case, the conjugate word $w$ is a palindrome.
	Using this result, Proposition~\ref{cyclic} and Item $(2)$ of Proposition~\ref{conjugacy}, we deduce that every morphism in Class $P_{ret}$ is conjugated to an acyclic morphism in Class $P$.
	
\medskip
	The converse is not true.
	For example, the morphism given by
	$$
	\varphi:
	\left\{
	\begin{array}{l}
		{\tt 0} \mapsto {\tt 010} \, {\tt 101} \\
		{\tt 1} \mapsto {\tt 010} \, {\tt 1001}
	\end{array}
	\right.
	$$
	belongs to Class $P$ and  is acyclic.
	In particular,
	$$
	\varphi_R:
	\left\{
	\begin{array}{l}
		{\tt 0} \mapsto {\tt 010101} \\
		{\tt 1} \mapsto {\tt 0101010}
	\end{array}
	\right.
	, \quad \quad
	\varphi_L:
	\left\{
	\begin{array}{l}
		{\tt 0} \mapsto {\tt 101010} \\
		{\tt 1} \mapsto {\tt 0101010}
	\end{array}
	\right.
	$$
	and the conjugate word is $w = {\tt 0101010}.$
	Nevertheless, in $\varphi_R({\tt 0})w$ the word $w$ occurs four times.
	Therefore, no conjugate of $\varphi$ belongs to Class $P_{ret}$.
	
		\section{Marked morphisms}
	\label{sec:marked}
	
	In our attempt to find morphisms preserving richness we will restrict to morphisms for which "desubstitution" is easy.
	We use marked morphisms.
	Our definition is slightly more general than the definition used by Frid in~\cite{Fr99}.
	
	\begin{defi}
		\label{marked}
		An acyclic morphism $\varphi$ is \emph{right marked} (resp. \emph{left marked}) if the mappings $\rho_\varphi$ and $\lambda_\varphi$ are injective on $\A$.
		
		A morphism is \emph{marked} if it is both right marked and left marked.
		
		A marked morphism is called \emph{well-marked} if both $\rho_\varphi$ and $\lambda_\varphi$ are the identity map on $\A$.
	\end{defi}
	
	The following proposition is a direct consequence of Item (2) in Proposition~\ref{conjugacy}.
	
	\begin{prop}
		\label{pro:rightleft}
		Let $\varphi$ be a morphism in Class $P_{ret}$.
		If $\varphi$ is right marked, then it is left marked too.
	\end{prop}
	
	\begin{lem}
		If $\varphi \in P_{ret}$ is right marked, then there exists a positive integer $k$ such that $\varphi^k$ is well-marked.
	\end{lem}
	\begin{proof}
		The statement follows simply from the following three facts:
		\begin{itemize}
			\item[1)] The mappings $\lambda_\varphi$ and $\rho_\varphi$ are permutations on $\A$.
			\item[2)] If $\pi$ is a permutation on $d$ elements, then for $k = d!$, the power $\pi^k$ is the identity.
			\item[3)] Class $P_{ret}$ is closed under composition (see~\cite{BaPeSt11}).
		\end{itemize}
	\end{proof}

\eject
	
	\begin{example}
		The morphism $\sigma$ in Example~\ref{ex:sigma2} is well marked.
		On the other hand, the Tribonacci morphism $\tau$ seen in Example~\ref{ex:varphi}, which is in Class $P_{ret}$ with marker ${\tt 0}$, is marked but not well marked.
		Indeed, one has $\tau = \tau_R$ and
		$$
		\tau_L:
		\left\{
		\begin{array}{l}
			{\tt 0} \mapsto {\tt 10} \\
			{\tt 1} \mapsto {\tt 20} \\
			{\tt 2} \mapsto {\tt 0}
		\end{array}
		\right. .
		$$
		On the other hand one can check that the morphism $\tau^3$ is well-marked.
		Indeed, one has
		$$
		\tau^3_R:
		\left\{
		\begin{array}{l}
			{\tt 0} \mapsto {\tt 0102010} \\
			{\tt 1} \mapsto {\tt 010201} \\
			{\tt 2} \mapsto {\tt 0102}
		\end{array}
		\right.
		\quad \mbox{and} \quad
		\tau^3_L:
		\left\{
		\begin{array}{l}
			{\tt 0} \mapsto {\tt 0102010} \\
			{\tt 1} \mapsto {\tt 102010} \\
			{\tt 2} \mapsto {\tt 2010}
		\end{array}.
		\right.
		$$
	\end{example}
	
	\begin{lem}
		\label{occurrence}
		Let $\uu \in \A^\N$ and $\varphi \in P_{ret}$ be a marked morphism over $\A$ with marker $w$, such that $\varphi = \varphi_R$.
		Let $a, b \in \A$ such that $b = \rho_\varphi(a)$.
		\begin{itemize}
			\item The factor $wb$ occurs in $\varphi(\uu)$ only as a prefix of $\varphi(a) w = w \varphi_L(a)$.
			\item The factor $bw$ occurs in $\varphi(\uu)$ only as a suffix of $\varphi(a)w = w \varphi_L(a)$.
		\end{itemize}
	\end{lem}

	\begin{proof}
		Let us set $\varphi(\uu) = b_0 b_1 b_2 \cdots$ and let us suppose that $i$ is an occurrence of the word $wb$ in $\varphi(\uu)$.
		Let us consider a $j > i$ such that $b_i b_{i+1} \cdots b_{j-1}$ is a return word to $w$ in $\varphi(\uu)$.
		By Item (2) of Lemma~\ref{lem:wprefix}, the index $i$ is an occurrence of $\varphi(c) w$ for some $c \in \A$.
		Since $\varphi(c) w = w \varphi_L(c)$ is a palindrome, then $b = \lambda_\varphi(c)$.
		By Item (2) of Proposition~\ref{conjugacy}, we have $b = \rho_\varphi(c)$.
		Since $\rho_\varphi$ is injective, we have $c = a$.
		
		Now, let us assume that $i-1$ is an occurrence of the word $bw$ in $\varphi(\uu)$.
		Since $\varphi(\uu)$ has a prefix $w$, there exists an index $j < i$ such that $b_j b_{j+1} \cdots b_{i-1}$ is a return word to $w$ in $\varphi(\uu)$.
		By the same argument as above, $i - |\varphi(c)|$ is an occurrence of $\varphi(c) w$ in $\varphi(\uu)$ for some $c \in \A$.
		Obviously, $b = \rho_\varphi(c)$
		and thus $c = a$.
	\end{proof}
	
	\begin{example}
		\label{ex:occurrences}
		Let $\tau = \tau_R$ and ${\bf t}$ be as in Example~\ref{ex:varphi} and let $w_\tau = {\tt 0}$ be the marker of $\tau$.
		Since $\rho_\tau({\tt 1}) = {\tt 2}$, it follows from Lemma~\ref{occurrence} that ${\tt 0 2}$ occurs in ${\bf t}$ only as a prefix of ${\tt 02 0}$ and that ${\tt 20}$ occurs in ${\bf t}$ only as a suffix of ${\tt 02 0}$.
	\end{example}

	\begin{proposition}
		\label{edges}
		Let $\varphi: \A^* \to \A^*$ be a marked morphism in Class $P_{ret}$ such that $\varphi = \varphi_R$ and $w$ be its marker.
		Let $\uu \in \A^\N$ and $a, b \in \A$.
		Then
		\begin{equation}
		a u b \in \L(\uu)
			\; \Longleftrightarrow \;
			\rho_\varphi(a) \, \varphi(u) w \, \rho_\varphi(b) \in \L(\varphi(\uu)). 	\label{edges1}
		\end{equation}
	\end{proposition}
	\begin{proof}
		Let us first suppose that $a u b \in \L(\uu)$.
		By Item (3) of Lemma~\ref{lem:wprefix}, $\varphi(a) \varphi(u) \varphi(b) w \in \L(\varphi(\uu))$.
		By definition $\varphi(a)$ has $\rho_\varphi(a)$ as a suffix.
		Since $\varphi = \varphi_R$, we can apply Item (2) of Proposition~\ref{conjugacy} to get $\varphi(b) w = w \varphi_L(b)$.
		Therefore, $\varphi(b)w$ starts with $w \lambda_\varphi(b) = w \rho_\varphi(b)$.
		Hence, $\varphi(a) \varphi(u) \varphi(b) w$ has a factor $\rho_\varphi(a) \varphi(u) w \rho_\varphi(b)$.
		
		To show the opposite implication, let us assume that $i - 1$ is an occurrence in $\varphi(\uu)$ of the factor $\rho_\varphi(a) \varphi(u) w \rho_\varphi(b) \in \L(\varphi(\uu))$.
		By the second assertion of Lemma~\ref{occurrence}, $i - |\varphi(a)|$ is an occurrence of the factor $\varphi(a) \varphi(u) w \rho_\varphi(b)$.
		By the first assertion of the same lemma, $i - |\varphi(a)|$ is an occurrence of the factor $\varphi(a) \varphi(u) \varphi(b) w = \varphi(aub) w$.
		Since $w$ is a prefix and a suffix of $\varphi(aub) w$, Corollary~\ref{preimage} implies $aub \in \L(\uu)$.
	\end{proof}
	
	The next result describes the relation between the set $\B(\uu)$ of bispecial factors in $\uu$ and the set $\B(\varphi(\uu))$ of bispecial factors in $\varphi(\uu)$, when $\varphi \in P_{ret}$.
	Given a permutation $\pi$ on $\A$ and a set $S \subset \A$ we denote  $\pi(S) = \left\{ \pi(s) \; : \; s \in S \right\}$.
	
	\begin{coro}
		\label{ImagesOfBSfactor}
		Let $\varphi: \A^* \to \A^*$ be a marked morphism in Class $P_{ret}$ such that $\varphi = \varphi_R$ and $w$ be its marker.
		Let $\uu \in \A^\N$.
		\begin{enumerate}
			\item If $v \in \B(\varphi(\uu))$ and $w$ is a factor of $v$, then there exists $u \in \B(\uu)$ such that $v = \varphi(u) w$.
			\item For every $u \in \L(\uu)$ we have
			$\rho_\varphi \left(L_{\uu}(u) \right) = L_{\varphi(\uu)}(\varphi(u)w)$,
			$\rho_\varphi \left( R_{\uu}(u) \right) = R_{\varphi(\uu)}(\varphi(u) w)$, and
			$\rho_\varphi \left (B_{\uu} (u) \right) = B_{\varphi(\uu)} (\varphi(u) w)$.
			\item $u \in \B(\uu)$ if and only if $\varphi(u)w \in \B(\varphi(\uu))$.
			\item $b_{\uu}(u) = b_{\varphi(\uu)}( \varphi(u)w )$ for every $u \in \L(\uu)$.
			\item $\pext{\uu}{u} =\pext{\varphi(\uu)}{\varphi(u)w}$ for every palindromic $u \in \L(\uu)$.
		\end{enumerate}
	\end{coro}
	
\begin{proof}

\vspace*{-4mm}
		\textbf{Item (1).}
		Let $v$ be a bispecial factor in $\L(\varphi(\uu))$ and $w$ a factor of $v$.
		There exist two letters $b_1$ and $b_2$, $b_1 \ne b_2$ such that $v b_1, vb_2 \in \L(\varphi(\uu))$.
		We show that $w$ is a suffix of $v$.
		Assume the contrary, i.e., that $v = ywx$, with $x, y \in \A^*$, such that the factor $w$ occurs in $wx$ only once and $x \ne \varepsilon$.
		Let $b$ be the first letter of $x$.
		By Lemma~\ref{occurrence} and the fact that both $w x b_1, w x b_2 \in\ \L(\varphi(\uu))$, we have that $\varphi(a)w$ is a prefix of $w x$, where $b = \rho_\varphi(a)$, a contradiction with the definition of $x$.
		Therefore $w$ is a suffix of $v$.
		By the same argument, $w$ is also a prefix of $v$.
		Corollary~\ref{preimage} implies that $v = \varphi(u)w$.
		
\medskip
		\textbf{Item (2).}
		Proposition~\ref{edges} says that $(a,b)$ is a bi-extension of $u$ in $\L({\uu})$ if and only if $\left( \rho_\varphi(a), \rho_\varphi(b) \right)$ is a bi-extension of $\varphi(u) w$ in $\L \left( \varphi(\uu) \right)$.
		
\medskip
		\textbf{Item (3).}
		A consequence of the previous item.
		
\medskip
		\textbf{Item (4).}
		A consequence of the previous item.
		
\medskip
		\textbf{Item (5).}
		A consequence of Equation~(3) and Item (3) of Lemma~\ref{palindromPreserved}.
	\end{proof}
	
	\begin{example}
		Let $\tau = \tau_R$, ${\bf t}$ and $w_\tau$ be as in Example~\ref{ex:occurrences}.
		One has $\rho_\tau = ({\tt 1, 2, 0})$.
		Let us consider the word $u = {\tt 2010}$.
		One has
		$$
		L_{\bf t}(u) = \left\{ {\tt 0} \right\},
		\quad
		R_{\bf t}(u) = \left\{ {\tt 0, 1, 2} \right\},
		\quad
		B_{\bf t}(u) = \left\{ ({\tt 0}, {\tt 0}), ({\tt 0}, {\tt 1}), ({\tt 0}, {\tt 2}) \right\}
		$$
		and
		$$
		L_{\bf t}(\tau(u) {\tt 0}) = \left\{ {\tt 1} \right\},
		\quad
		R_{\bf t}(\tau(u) {\tt 0}) = \left\{ {\tt 1, 2, 0} \right\},
		\quad
		B_{\bf t}(\tau(u) {\tt 0}) = \left\{ ({\tt 0}, {\tt 1}), ({\tt 0}, {\tt 2}), ({\tt 0}, {\tt 0}) \right\}.
		$$
	\end{example}
	
	\begin{remark}
		\label{UpToConjugate}
		Let $\uu \in \A^\N$ be recurrent and $\varphi$ and $\psi$ be two mutually conjugate morphisms over $\A$.
		Clearly, the languages of $\varphi(\uu)$ and $\psi(\uu)$ coincide.
		Since the palindromic richness can be seen as a property of a language and not of an infinite word itself, it is enough to examine richness for one of these languages.
	\end{remark}
	
	\begin{coro}
		\label{EasyDirection}
		Let $\varphi: \A^* \to \A^*$ be a marked morphism in Class $P_{ret}$ and $\uu \in \A^\N$ such that its language is closed under reversal.
		If $\varphi(\uu)$ is rich, then $\uu$ is rich.
	\end{coro}
	\begin{proof}
		Let us recall that if a language is closed under reversal, then it is recurrent.
		By Remark~\ref{UpToConjugate}, we may assume without loss of generality that $\varphi = \varphi_R$.
		Let $w$ denote the marker associated to $\varphi$.
		By Corollary~\ref{preimage}, the language of $\varphi(\uu)$ is closed under reversal as well.
		Therefore, to demonstrate the richness of $\uu$ we may use Theorem~\ref{OnlyBSfactors}.
		If $u\in \B(\uu)$, then by Item (1) of Corollary~\ref{ImagesOfBSfactor}, $\varphi(u)w \in \B(\varphi(\uu))$.
		Since $\varphi(\uu)$ is rich, we have that $b_{\varphi(\uu)}( \varphi(u)w ) = \pext{\varphi(\uu)}{\varphi(u)w} - 1$ if $\varphi(u)w$ is a palindrome, and $b_{\varphi(\uu)}( \varphi(u)w ) = 0$ otherwise.
		Recall that by Lemma~\ref{palindromPreserved}, $u$ is a palindrome if and only if $\varphi(u)w$ is a palindrome.
		Therefore, Items (4) and (5) of Corollary~\ref{ImagesOfBSfactor} imply that $b_{\uu}( u) = \pext{\uu}{u} - 1$ if $u$ is a palindrome, and $b_{\uu}( u) = 0$ otherwise.
		Consequently, $\uu$ is rich.
	\end{proof}
	
	\begin{remark}
		\label{rem:elementaryAR}
		Let us point out that the elementary Arnoux-Rauzy morphism  $\psi_a$, with $a \in \A$, belongs to Class $P_{ret}$ and $w = a$ is its associated marker.
		Moreover, $\psi_a$ is well marked.
		Hence Lemma 5.3. in \cite{GlJuWiZa09} is a special case of the previous corollary for $\varphi = \psi_a$.
	\end{remark}

	\section{Morphisms preserving richness}
	\label{sec:main}
	
	Our aim is to study the implication opposite to Corollary~\ref{EasyDirection}, i.e., for which morphism $\varphi$ in Class $P_{ret}$, richness of $\uu$ forces richness of $\varphi(\uu)$.
	Behaviour of morphisms from this class may differ according to the infinite word we apply the morphism to.
	This is testified by the following example (see~\cite{BaPeSt11}).
	
	\begin{example}
		Let us consider the morphism $\zeta: \{ {\tt 0,1,2} \}^* \mapsto \{ {\tt 0,1} \}^*$ defined as
		$$
		\zeta:
		\left\{
		\begin{array}{l}
			{\tt 0} \mapsto {\tt 0100} \\
			{\tt 1} \mapsto {\tt 01011} \\
			{\tt 2} \mapsto {\tt 010111}
		\end{array}.
		\right.$$
		This morphism is in Class $P_{ret}$.
		It is proved in~\cite[Lemma 5.6]{BaPeSt11} that $\zeta$ maps a rich word (over the ternary alphabet $\{ {\tt 0,1,2} \}$) to a word with infinite palindromic defect (over the binary alphabet $\{ {\tt 0,1} \}$).
		Moreover, the same morphism also maps the Tribonacci word (which is rich) to a rich word.
	\end{example}
	
	Starosta studied in~\cite {St16} morphic images of episturmian words by morphisms of Class $P_{ret}$.
	It is known that episturmian words are rich (see, for instance,~\cite{GlJu09}).
	
	\begin{thm}(\cite{St16})
		Let $\uu$ be an episturmian word and $\varphi: \A^* \to \B^*$ be a morphism of Class $P_{ret}$.
		The word $\varphi(\uu)$ is almost rich.
	\end{thm}
	
	In \cite[Corollary 6.3]{GlJuWiZa09} it is proved that the image of each rich word by an Arnoux-Rauzy morphism is rich as well.
	Here we provide an alternative proof of this statement.
	
	\begin{thm}(\cite{GlJuWiZa09})
		\label{richAR}
		Let $\varphi$ be an Arnoux-Rauzy morphism over the alphabet $\A$ and $\uu \in \A^\N $ such that its language is closed under reversal.
		Then $\uu$ is rich if and only if $\varphi(\uu)$ is rich.
	\end{thm}
	\begin{proof}
		Since the Arnoux-Rauzy monoid is generated by elementary morphisms, let us prove the result for these morphisms.
		
		Let us first consider the case $\varphi =\psi_a$ for a certain letter $a \in \A$.

\medskip		
		By Remark~\ref{rem:elementaryAR}, if $\psi_a(\uu)$ is rich, then $\uu$ is rich as well.
		To prove the opposite implication let us assume that $\uu$ is rich and let us prove that for all bispecial words in $\mathcal{L}(\psi_a({\uu}))$ Equation~(1) is satisfied.
		The empty word $\varepsilon$ is clearly a palindrome.
		The only palindromic extension of the empty word in $\L(\psi_a({\uu}))$ is $aa$, and $b_{\psi_a({\uu})}(\varepsilon) = 1 - 1 = 0$.
		The only bispecial factor of length 1 is $a$, that is clearly a palindrome.
		Using Items (4) and (5) of Corollary~\ref{ImagesOfBSfactor} we have:
		$$
		\pext{\psi_a({\uu})}{a}
		\; = \;
		\pext{\bf u}{\varepsilon}
		\; = \;
		b_{\bf u}(\varepsilon) + 1
		\; = \;
		b_{\psi_a({\uu})} (a) + 1.
		$$
		Let us now consider a bispecial word $v \in \L(\psi({\uu}))$ of length at least $2$.
		Since the only left special (resp. right special) letter is $a$, we can write $v = a v' a$ for a certain word $v'$.
		Moreover $av' = \psi_a(u)$ for a certain word $u \in \L({\uu})$, hence $u = \psi_a(w) a$ and, by Item (3) of Corollary~\ref{ImagesOfBSfactor}, $u \in \B({\uu})$.
		Using Item (3) of Lemma~\ref{palindromPreserved}, we have that $v$ is a palindrome if and only if $u$ is a palindrome.
		Using the same reasoning as before we find that whenever $v$ is a palindrome, then
		$$
		\pext{\psi_a({\uu})}{v}
		\; = \;
		\pext{\uu}{u}
		\; = \;
		b_{\uu}(u) + 1
		\; = \;
		b_{\psi_a({\uu})}(u) + 1,
		$$
		while whenever $v$ is not a palindrome, then
		$$
		b_{\psi_a({\uu})}(v)
		\; = \;
		b_{\uu}(u)
		\; = \;
		0.
		$$
		Thus, according to Theorem~\ref{OnlyBSfactors}, $\psi_a(\uu)$ is rich.

\medskip		
		Let us now consider the case $\varphi = \overline{\psi}_a$ for a certain letter $a \in \A$.
		We note that $\overline{\psi}_a$ is conjugated to $\psi_a$.
		Due to Remark~\ref{UpToConjugate} we have that
		$$
		\overline{\psi}_a(\uu) \; \mbox{ is rich}
		\quad \Longleftrightarrow \quad
		\psi_a(\uu) \; \mbox{ is rich }
		\quad \Longleftrightarrow \quad
		\uu \; \mbox{ is rich }
		$$
		
		Finally, let us consider a morphism permuting letters.
		Since the definition of richness is based on counting palindromes occurring in a factor of $\uu$, permutations of letters do not change the number of palindromes and the statement of the theorem is trivially verified.
	\end{proof}
	
	In the case of the binary alphabet the set of morphisms in Class $P_{ret}$ which preserve richness can be extended even more.
	Indeed, in a binary alphabet, every acyclic morphism $\psi$ is marked and moreover a morphism in Class $P_{ret}$ is acyclic by Proposition~\ref{cyclic}.
	
	\begin{lem}
		\label{lem:short}
		Let $w$ be the marker of a binary morphism $\varphi = \varphi_R$ in Class $P_{ret}$ and $\uu \in \{ {\tt 0,1} \}^\N$ be a recurrent word.
		Assume that $\uu$ and $\varphi({\tt 01})w$ are rich.
		Each bispecial factor of $\varphi(\uu)$ which does not contain the marker $w$ satisfies Condition~(1) of Theorem~\ref{OnlyBSfactors}.
	\end{lem}
	\begin{proof}
		Let $v$ be a bispecial factor of $\varphi(\uu)$ which does not contain $w$ as a factor.
		By Lemma~\ref{lem:wprefix}, we deduce that $v$ is a factor of $c^{-1} \varphi({\tt 0}) w c^{-1}$ or $c^{-1} \varphi({\tt 1}) w c^{-1}$, where $c ={\rm  Fst}(w) = {\rm Lst}(w)$.
		Since $\varphi({\tt 0})w$ and $\varphi({\tt 1})w$ are palindromes, we have
		\begin{equation}
			\label{star1}
			xvy \in \L \left( \varphi({\tt 01})w \right)
			\quad \Longleftrightarrow \quad
			y \overline{v} x \in \L \left( \varphi({\tt 01})w \right)
			\qquad
			\text{for } x,y \in \{ {\tt 0,1} \}\,,
		\end{equation}
		and
		\begin{equation}
			\label{star2}
			xvy \in \L \left( \varphi({\tt 01})w \right)
			\quad \Longleftrightarrow \quad
			xvy \in \L \left( \varphi(\uu) \right)
			\qquad
			\text{for } x,y \in \{ {\tt 0,1} \}\,.
		\end{equation}
		Let us consider the set of bi-extensions $\mathfrak{B} = \{xvy \; : \; (x,y) \in B_{\varphi(\uu)}(v) \}$.
		By Equation~(5), we have $\mathfrak{B} \subset \L \left( \varphi({\tt 01})w \right)$.
		Let us show that $v$ satisfies condition~(1) of Theorem~\ref{OnlyBSfactors}, by considering the bilateral order of $v$.
		Note that, since we are in a binary alphabet, the only possible values for $b_{\varphi(\uu)}(v)$ are $0, 1$ and $-1$.
\begin{description}
			\item[Case $b_{\varphi(\uu)}(v) = 0$] \quad
			In this case we have $\# \mathfrak{B} = 3$.
			If $v$ is not a palindrome, then Equation~(1) is trivially fulfilled.
			If $v$ is a palindrome, then Equation~(4)  implies that $\mathfrak{B} = \{ ava, avb, bva\}$ for $a,b \in \{ {\tt 0,1} \}$, and $a \ne b$.
			Therefore, $\pext{\varphi(\uu)}{v} = 1$ and Equation~(1) is satisfied.
			
			\item[Case $b_{\varphi(\uu)}(v) = 1$] \quad
			In this case we have $\mathfrak{B} = \{ {\tt 0}v{\tt 1}, {\tt 1}v{\tt 0}, {\tt 0}v{\tt 0}, {\tt 1}v{\tt 1} \}$.
			If $v$ is a palindrome, then we have $\pext{\varphi(\uu)}{v} = 2$ and Equation~(1)  is fulfilled.
			If $v$ is not a palindrome, then Equation~(4)  implies that ${\tt 0}v{\tt 0}, {\tt 1}v{\tt 1}, {\tt 0}\overline{v}{\tt 1}$ 
            and ${\tt 1}\overline{v}{\tt 0}$ occur in $\varphi({\tt 01})w$.
			By Lemma~\ref{Montreal}, this is a contradiction with richness of $\varphi(01)w$.
			Therefore, $b_{\varphi(\uu)}(v)$ cannot be positive for a non-palindromic $v$.
			
			\item[Case $b_{\varphi(\uu)}(v)(v) = -1$] \quad
			In this case $\#\mathfrak{B} = 2$.
			First assume that $\mathfrak{B} = \{ {\tt 0}v{\tt 1}, {\tt 1}v{\tt 0} \}$.
			If $v$ is a palindrome, then $\pext{\varphi(\uu)}{v} = 0$ and Equation~(1)  is fulfilled.
			If $v$ is not a palindrome, then by Lemma~\ref{alternate} occurrences of $v$ and $\overline{v}$ alternate in $c^{-1}\psi({\tt 01})wc^{-1}$ and occurrences of ${\tt 0}v{\tt 1}$ and ${\tt 1}\overline{v}{\tt 0}$ alternate in $\varphi({\tt 01})w$.
			Item (2) of Lemma~\ref{alternate} implies that ${\tt 0}v{\tt 1}$ and ${\tt 1}v{\tt 0}$ cannot occur simultaneously in $\varphi({\tt 01})w$ which leads to a contradiction.
			Therefore, $\mathfrak{B} \ne \{{\tt 0}v{\tt 1}, {\tt 1}v{\tt 0}\}$ for a non-palindromic $v$.
			Moreover, $\mathfrak{B} \ne \{{\tt 0}v{\tt 0}, {\tt 1}v{\tt 1}\}$.
			Indeed, by Item (2) of Lemma~\ref{alternate} we can not have simultaneous occurrences of ${\tt 0}v{\tt 0}$ and ${\tt 1}v{\tt 1}$ in $\varphi({\tt 01})w$ for both palindromic and non-palindromic $v$.
		\end{description}

\vspace*{-6mm}
	\end{proof}
	
	\begin{example}
		Let ${\bf f}$ be as in Example~\ref{ex:varphi} and $\varphi_1$ and $w_1$ as in Example~\ref{ex:psi1}.
		Let $\vv = \varphi_1({\bf f})$.
		The word $\varphi_1({\tt 01}) w_1 = {\tt 1110 1 111}$ is rich.
		The only bispecial factors of $\vv$ which do not contain $w_1$ as a factor are the palindromes $\varepsilon, {\tt 1}$ and ${\tt 11}$.
		One can check that $b_\vv(\varepsilon) = b_\vv({\tt 1}) = b_\vv({\tt 11}) = 0$ and that $\pext{\vv}{\varepsilon} = \pext{\vv}{\tt 1} = \pext{\vv}{\tt 11} = 1$.
	\end{example}
	
	\begin{lem}
		\label{lem:psiurich}	
		Let $\uu$ be a fixed point of a primitive morphism $\varphi: \{ {\tt 0,1} \}^* \to \{ {\tt 0,1} \}^*$  conjugated to a morphism in Class $P_{ret}$.
		Then $\L(\uu)$ is closed under reversal.
	\end{lem}
	\begin{proof}
		By Item (3) of Lemma~\ref{palindromPreserved} and by Item (3) of Lemma~\ref{lem:wprefix}, a fixed point of a morphism in Class $P_{ret}$ contains infinitely many palindromes.
		Since $\varphi$ is primitive, $\L(\uu)$ is uniformly recurrent, i.e., for each $n$ there exists  an integer number $R(n)$ such that every factor of length at least $R(n)$ contains  all factors of $\L(\uu)$ of length $n$.
		In particular, if $u \in \L(\uu)$ and $|u| = n$, then $u$ is a factor of every palindrome $p \in \L(\uu)$  of length at least $R(n)$.
		Hence, $\overline{u}$ belongs to $\L(\uu)$ as well.
	\end{proof}
	
	\begin{thm}
		\label{theo:psiurich}
		Let $\varphi: \{ {\tt 0,1} \}^* \to \{ {\tt 0,1} \}^*$ be a morphism conjugated to a morphism in Class $P_{ret}$.
		Let $w$ be the marker associated to $\varphi_R$.
		Assume that the finite word $\varphi_R({\tt 01})w$ is rich.
		\begin{enumerate}
			\item If $\uu \in \{ {\tt 0,1} \}^\N$ is recurrent and rich, then $\varphi(\uu)$ is rich.
			\item If $\uu \in \{ {\tt 0,1} \}^\N$ is a fixed point of $\varphi$ and $\varphi$ is primitive, then $\uu$ is rich.
		\end{enumerate}
	\end{thm}
	\begin{proof}
		As we have already mentioned, for a recurrent word $\uu$ the language of $\varphi_R(\uu)$ and of $\varphi(\uu)$ coincide.
		In the proof of the second assertion we use the fact that fixed points of mutually conjugate primitive morphisms have the same language, which is moreover uniformly recurrent.
		Therefore, we may assume without loss of generality that $\varphi = \varphi_R$.
		To prove both items we need to use Theorem~\ref{OnlyBSfactors}.
		Note that if $\uu \in \{ {\tt 0,1} \}^\N$ is recurrent and rich, then $\L(\uu)$ is closed under reversal.
		The same is true when $\uu \in \{ {\tt 0,1} \}^\N$ is a fixed point of a primitive morphism by Lemma~\ref{lem:psiurich}.
		
		To prove Item $(1)$ let us assume richness of $\uu$.
		Let $v \in \B(\varphi(\uu))$.
		If $v$ does not contain $w$ as its factor, then by Lemma \ref{lem:short} $v$ satisfies Condition~(1) of Theorem~\ref{OnlyBSfactors}.
		If $w$ occurs in $v$, then by Item (1) of Proposition~\ref{ImagesOfBSfactor} there exists $u \in \B(\uu)$ such that $v = \varphi(u)w$.
		Since $\uu$ is rich, $u$ satisfies Equation~(1).
		By Items (4) and (5) of the Proposition~\ref{ImagesOfBSfactor} and by Item (3) of Lemma~\ref{palindromPreserved}, $v$ satisfies the same condition.
		Consequently, $\varphi(\uu)$ is rich.
		
		To prove Item (2) let us consider $\uu$ such that $\varphi(\uu) = \uu$.
		Letting $c ={\rm Fst}(w)$, we have $\uu = \lim \varphi^n(c)$.
		In accordance with~\cite{Kl12}, a bispecial factor of $\uu$ which does not contain $w$ will be called \emph{initial}.
		Note that there are only a finite number of initial bispecial factors since all long enough words in $\L(\uu)$ contain $w$ as a factor.
		We have already showed that every initial factor satisfies Equation~(1).
		Item (1) of Proposition~\ref{ImagesOfBSfactor} allows us to define a mapping $\Phi: \B(\uu) \to \B(\uu)$ defined as $\Phi(u) := \varphi(u) w$.
		Item (2) of the same proposition says that for each non-initial bispecial factor $v' \in \B(\uu)$ there exists $n \in \N$ and an initial bispecial factor $v$ of $\uu$ such that $v' = \Phi^n(v)$.
		By Items (3) and (4) of the same proposition and by Lemma~\ref{palindromPreserved}, each bispecial factor of $\uu$ fulfills the condition of Theorem~\ref{OnlyBSfactors} and thus $\uu$ is rich.
	\end{proof}
	
	\begin{example}
		\label{ex:sigma0}
		Let $\sigma, \sigma_R$ and $w$ be as defined in Example~\ref{ex:sigma2} and ${\bf f}$ be the Fibonacci word seen in Example~\ref{ex:varphi}.
		One can easily check that the words $\sigma_R({\tt 01}) w = {\tt 10 1011 101}$ is rich and that $\sigma$ is primitive.
		Moreover, it is well known that ${\bf f}$ is rich and recurrent.
		Thus, according to point (1) of Theorem~\ref{theo:psiurich}, the infinite word $\sigma({\bf f})$ is rich.
		Moreover, according to point (2) of the same theorem the fixed point
		$$
		\sigma({\tt 0})^\omega = {\tt 0 1 0111 01 0111 0111 0111 01 0111 01 0111 0111 0111 01} \cdots,
		$$
		is rich (as well as $\sigma_R({\tt 1})^\omega$).
	\end{example}
	
	\begin{coro}
		Let $\varphi$ be a binary morphism in Class $P_{ret}$ and $\uu \in\{ {\tt 0,1} \}^\N$ be a non-unary recurrent word.
		If $\varphi(\uu)$ is rich, then $\varphi(\vv)$ is rich for every recurrent rich word $\vv \in \{ {\tt 0,1} \}^\N$.
	\end{coro}
	\begin{proof}
		Let $w$ be the marker of $\varphi_R$.
		Since $\uu$ is non-unary recurrent word, then ${\tt 01} \in \L(\uu)$.
		Remark~\ref{UpToConjugate} and Item (4) of Lemma~\ref{lem:wprefix} imply that $\varphi_R({\tt 01})w$ is a factor of $\uu$ and thus $\varphi_R({\tt 01})w$ is rich.
		The statement follows by Item (1) of the previous theorem.
	\end{proof}
	
	\begin{example}
		Let $\phi$ be as in Example~\ref{ex:varphi} and $\sigma$ as in Example~\ref{ex:sigma0}.
		Since ${\bf f = \phi({\bf f})}$ is recurrent and rich, we have that
		$$
		\phi(\sigma({\tt 0)}^\omega) = {\tt 01 0 01 0 0 0 01 0 01 0 0 0 01 0 0 0 01 0 0 0 01 0 01 0 0 0 01 0 01 0 0 0 01 0 0 } \cdots
		$$
		is rich.
		Moreover, images of rich periodic words by $\phi$ (or of any other Sturmian morphism) are rich as well.
	\end{example}

	\section{Construction of new rich words }
	\label{sec:derived}
	
	Theorems~\ref{richAR} and~\ref{theo:psiurich} enable us to construct new rich words from known ones.
	We give several examples.

\medskip	
	On alphabets of cardinality $d$, two classes of infinite words are known to be rich:
	\begin{itemize}
		\item episturmian words and
		\item the so-called $d$-\emph{iet words}, that is words coding $d$-interval exchange transformation under the symmetric permutation $(d, d-1, \ldots, 2, 1)$.
	\end{itemize}
	Applying an arbitrary Arnoux-Rauzy morphism to a $d$-iet word gives a new rich word which does not belong to any of the two classes mentioned above.
	
	Although both our theorems consist in applying a richness-preserving morphism to infinite words, we can apply them to finite words as well.
	The construction of infinite words J. Vesti provided in the proof of Proposition~\ref{ExtensionToPeriodic} is based on repetitive application of palindromic closure.
	Therefore, the languages of infinite words he constructed are closed under reversal and thus we can apply our theorems to factors of these words to get new finite rich words.
	
	In~\cite{GuShSh16}, Guo, Shallit and Shur gave a lower bound on the number of binary rich words of length $n$.
	They proved that every word of the form
	$
	{\tt 0}^{a_1} {\tt 1}^{b_1} {\tt 0}^{a_2} {\tt 1}^{b_2} \cdots {\tt 0}^{a_k} {\tt 1}^{b_k}
	$,
	where $a_1 \leq a_2 \leq \cdots \leq a_k$ and $b_1 \leq b_2 \leq \cdots \leq b_k$, is rich and thus they counted words of this form of length $n$.
	Mapping any of these rich words through a Sturmian morphism or a morphism from Theorem~\ref{theo:psiurich} gives us a new rich word.
	
	To give another type of examples of morphisms in Class $P_{ret}$ which have rich fixed points let us recall the definition of derived words as introduced by Durand in~\cite{Du98}.
	
	\begin{defi}
		Let $\RR{\uu}{x} = \{r_1,r_2, \ldots, r_k\}$ be the set of return words to a factor $x$ of a uniformly recurrent word $\uu$.
		Let $z$ be the shortest prefix of $\uu$ such that $\uu$ can be written in the form
		$
		\uu = z r_{d_0} r_{d_1} r_{d_2} \cdots
		$,
		where $r_{d_j} \in \RR{\uu}{x}$ for each $j \in \N$.
		The \emph{derived word} of $\uu$ to the factor $x$ is the infinite word $\dd_\uu(x) = d_0 d_1 d_2 \cdots$ over the alphabet of cardinality $\# \mathcal{R}_\uu(x) = k$.
	\end{defi}
	
	\begin{example}
		\label{Ex:PeriodDoubling}
		Let us consider the binary morphism
		$$
		\xi:
		\left\{
		\begin{array}{l}
			{\tt 0} \mapsto {\tt 11} \\
			{\tt 1} \mapsto {\tt 10}
		\end{array}
		\right.
		$$
		fixing the period doubling sequence
		$$
		\xx = {\tt 1011101010111011101} \cdots
		$$
		which is known to be rich (see~\cite{BlBrLaVu110}).
		The morphism $\xi$ does not belong to $P_{ret}$, nevertheless it is in Class $P$.
		Let us find the derived word of $\xx$ to the prefix $p = {\tt 1}$.
		This prefix has two return words, namely ${\tt r} = {\tt 10}$ and ${\tt s} = {\tt 1}$.
		We thus have
		$$
		\xx = \underbracket[1pt][2pt]{{\tt 10}}_{\tt r} \underbracket[1pt][2pt]{{\tt 1}}_{\tt s} \underbracket[1pt][2pt]{{\tt 1}}_{\tt s} \underbracket[1pt][2pt]{{\tt 10}}_{\tt r} \underbracket[1pt][2pt]{{\tt 10}}_{\tt r} \underbracket[1pt][2pt]{{\tt 10}}_{\tt r} \underbracket[1pt][2pt]{{\tt 1}}_{\tt s} \underbracket[1pt][2pt]{{\tt 1}}_{\tt s} \underbracket[1pt][2pt]{{\tt 10}}_{\tt r} \underbracket[1pt][2pt]{{\tt 1}}_{\tt s} \underbracket[1pt][2pt]{{\tt 1}}_{\tt s} \underbracket[1pt][2pt]{{\tt 10}}_{\tt r} \underbracket[1pt][2pt]{{\tt 1}}_{\tt s} \cdots
		$$
		and the derived word over the alphabet $\{ {\tt r,s} \}$ is
		$$
		{\bf d}_\xx(p) = {\tt rssrrrssrssrs} \cdots.
		$$
		By a result of Durand (\cite{Du98}), the derived word to a prefix is fixed by a substitution.
		Since $\xi({\tt r}) = {\tt 1011} = {\tt rss}$ and $\xi({\tt s}) = {\tt 10} = {\tt r}$, the substitution fixing ${\bf d}_\xx(p)$ is
		$$
		\eta:
		\left\{
		\begin{array}{l}
			{\tt r} \mapsto {\tt rss} \\
			{\tt s} \mapsto {\tt r}
		\end{array}.
		\right.
		$$
		One can easily check that $\eta : \{ {\tt r,s} \}^* \to \{ {\tt r,s} \}^*$ is in Class $P_{ret}$, with the marker $w = {\tt r}$. It is well-marked and the word $\eta({\tt rs})w = {\tt rssrr}$ is rich.
		Theorem~\ref{theo:psiurich} implies that the derived word to the prefix $p = {\tt 1}$ is rich as well.
	\end{example}
	
	What seen in the previous example is not surprising.
	Indeed in~\cite[Lemma 3.3]{HaVeZa16} it is stated that the derived word to the first letter of a rich word is rich as well.
	The same result for a palindromic factor of a rich word is contained in the proof of~\cite[Theorem 5.5.]{BaPeSt11}.
	
	\begin{lem}[\cite{HaVeZa16,BaPeSt11}]
		\label{againRich}
		Let $\uu$ be a rich infinite word.
		The derived word ${\bf d}_\uu(x) $ to any palindromic factor $x$ of $\uu$ is rich.
	\end{lem}
	
	The assumption that $x$ is palindromic in the previous lemma is necessary, as shown in the following remark.
	
	\begin{remark}
		In~\cite{MePeVu19}, the derived words to bispecial factors of complementary symmetric Rote sequences are described.
		The derived word to each bispecial factor is the coding of a three interval exchange transformation.
		If a bispecial factor $x$ is palindromic, then the interval exchange transformation is given by the symmetric permutation $(3,2,1)$.
		If a bispecial factor $x$ is non-palindromic, then the corresponding permutation is $(2,3,1)$.
		In the latter case the language of the derived word ${\bf d}_\uu(x)$ is not closed under reversal and contains only finitely many palindromes.
		Therefore, ${\bf d}_\uu(x)$ cannot be rich.
		In fact, the defect of ${\bf d}_\uu(x)$ is infinite.
	\end{remark}
	
	The matrix of a morphism $\varphi$ over an alphabet $\A$ of cardinality  $d =\#\A$ is defined as the matrix $M_{\varphi} \in \N^{d\times d}$ with entries
	$$
	\bigl(M_{\varphi}\bigr)_{ab} =
	\text{number of occurrences of the letter $a$ in $\varphi(b)$ for each $a,b \in \A$. }
	$$
	If the morphism $\varphi$ is primitive, then the Perron-Frobenius theorem (see for instance \cite{Meyer}) implies that  the spectral radius of $M_\varphi$, denoted by $\Lambda$, is a simple eigenvalue of $M_\varphi$.
	The eigenvalue $\Lambda$ is dominant, i.e., all other eigenvalues are in modulus strictly smaller than $\Lambda$.
	\medskip

	The dominant eigenvalue has an eigenvector of the form $(\rho_1, \rho_2, \ldots, \rho_d)$ with all positive entries $\rho_a>0$ and $\displaystyle \sum_{a \in \A} \rho_a = 1$.
	If $\uu = u_0 u_1 u_2 \cdots$ is a fixed point of the morphism $\varphi$, then $\rho_a$ equals the density of the letter $a$ in $\uu$, i.e.,
	$$
	\rho_a = \lim_{n \to \infty}\frac{\#\{i < n\, : \,  u_i  = a\}}{n}.
	$$
	
	Combining Lemma~\ref{againRich} with~\cite[Proposition 9]{Du2}, we get the following.
	
	\begin{lem}
		\label{lem:dxfixed}
		Let $\uu$ be a rich word fixed by a primitive substitution $\varphi$, $p$ be a palindromic prefix of $\uu$ and $\Lambda$ be the dominant eigenvalue of $M_\varphi$.
		The derived word ${\bf d}_\uu(p)$ is fixed by a primitive substitution $\eta$ in Class $P_{ret}$ such that the dominant eigenvalue of $M_\eta$ is $\Lambda$.
		In particular, the densities of letters in $\uu$ and ${\bf d}_\uu(p)$ (possibly over alphabets of different sizes) belong both to the algebraic field $\mathbb{Q}(\Lambda)$.
	\end{lem}
	
	Lemma~\ref{lem:dxfixed} says that the rich fixed points of $\varphi$ and $\eta$ are, in a way, not far from each other.
	In Example~\ref{Ex:PeriodDoubling}, the period-doubling substitution $\xi$ and the substitution $\eta$ have the same dominant eigenvalue $\Lambda = 2$.
	The densities of letters are $\rho({\tt 0}) = \tfrac13$, $\rho({\tt 1}) = \tfrac23$ in the period doubling sequence and $\rho({\tt r}) = \tfrac12$, $\rho({\tt s}) = \tfrac12$ in ${\bf d}_\uu(1)$.
	
	Let us show how Theorem~\ref{theo:psiurich} enables us, starting from a binary rich word $\uu$, to construct new rich words which differs substantially from $\uu$.
	We will use the fact that every aperiodic rich word has infinitely many palindromic bispecial factors (see \cite{LaPe16}).
	
	\subsection*{Construction:}

	Let $\uu$ be a binary rich word and $w$ a palindromic bispecial factor of $\uu$.
	Let $r_0$ and $r_1$ be two return words to $w$ in $\uu$ such that ${\rm Lst}(r_0)\ne {\rm Lst}(r_1)$ and $r_0r_1 \in \L(\uu)$.
	Then the substitution $\psi$ defined by $\psi({\tt 0}) = r_0$ and $\psi({\tt 1}) = r_1$ belongs to $P_{ret}$, its marker is $w$, $\psi = \psi_R$ and $\psi({\tt 0}) \psi({\tt 1}) w$ is rich.
	By Theorem~\ref{theo:psiurich}, every fixed point of $\psi$ is rich.
	
	\begin{example}
		Consider the fixed point $\vv$ of the substitution $\eta$ from Example~\ref{Ex:PeriodDoubling}, i.e.
		$$
		\vv = {\tt rssrrrssrssrssrrrssrrrssrrrs} \cdots.
		$$
		Just considering the short prefix of $\vv$ above, we can see that the palindromic factor ${\tt rr}$ has three return words, namely $r_0 = {\tt r}$, $r_1 = {\tt rrssrssrss}$, and $r_2 = {\tt rrss}$.
		Moreover, ${\rm Lst}(r_0) \ne {\rm Lst}(r_1)$ and ${\rm Lst}(r_0) \ne {\rm Lst}(r_2)$.
		Since $r_1 r_0$ and $r_2 r_0$ are factors of $\vv$, the above described construction gives two possible substitutions:
		$$
		\alpha:
		\left\{
		\begin{array}{l}
			{\tt r} \mapsto {\tt rrssrssrss} \\
			{\tt s} \mapsto {\tt r}
		\end{array}
		\right.
		\qquad \mbox{and} \qquad
		\beta:
		\left\{
		\begin{array}{l}
			{\tt r} \mapsto {\tt rrss} \\
			{\tt s} \mapsto {\tt r}
		\end{array}.
		\right.
		$$
		It is easy to verify that the dominant eigenvalue of $M_\alpha$ is $ \Lambda_\alpha =  2 + \sqrt{10}$, whereas $\Lambda_\beta = 1 + \sqrt{3}$.
		Therefore, fixed points of $\eta$, $\alpha$ and $\beta$ are rich and differ substantially.
		We can continue applying the same construction to the fixed points of the new substitutions.
	\end{example}

	\section{Comments and open questions}
	\label{sec:conclusions}
	
	Theorem~\ref {theo:psiurich} gives a simple criterion to decide whether a fixed point of a binary primitive morphism $\varphi$ in Class $P_{ret}$ is rich: namely, one has to check richness of the finite word $\varphi_R({\tt 01})w$.
	It would be helpful to find a similar criterion for morphisms over larger alphabets starting, for instance, with marked ternary morphisms in Class $P_{ret}$.
	The main technical obstacle is the absence of a ternary analogue of Lemma~\ref{Montreal}.
	
\medskip	
	On ternary alphabet, every word with language closed under reversal with factor complexity ${\mathcal{C}}(n) = 2n+1$, i.e., having exactly $2n+1$ factors of length $n$ for every integer $n \ge 0$, is rich (see~\cite{BaPeSt09}).
	This can be generalized to the class of \emph{dendric sets}, i.e., factorial sets such that the extension graph of every word in the set is acyclic and connected.
	Indeed, all dendric sets closed under reversal are rich (see~\cite{specular}).
	It would be useful to characterize all dendric sets which are closed under reversal.
	
\medskip	
	A morphism over an alphabet $\A$ is called \emph{tame} if it belongs to the submonoid generated by the permutations of letters of $\A$ and the morphisms $\alpha_{a,b}, \widetilde{\alpha}_{a,b}$ defined for $a,b \in \A$ with $a \ne b$ by
	$$
	\alpha_{a,b}(c) =
	\left\{
	\begin{array}{ll}
		ab	&	\mbox{if } c = a, \\
		c	&	\mbox{otherwise}
	\end{array}
	\right.
	\quad \mbox{and} \quad
	\widetilde{\alpha}_{a,b}(c) =
	\left\{
	\begin{array}{ll}
		ba	&	\mbox{if } c=a, \\
		c	&	\mbox{otherwise}
	\end{array}
	\right.
	$$
	(see, for instance~\cite{maximalbifixdecoding}, where the definition is given for morphisms of the free group instead of the free monoid).
	All Arnoux-Rauzy morphisms are tame, but the opposite is not true (the two classes coincide only when $\#\A \le 2$).
	As an example we can consider the morphism
	$$
	\psi = \pi_{({\tt c \, a \, b})} \circ \alpha_{{\tt a,b}} \circ \widetilde{\alpha}_{{\tt a,b}}
	\left\{
	\begin{array}{l}
		{\tt a} \mapsto {\tt aca} \\
		{\tt b} \mapsto {\tt a} \\
		{\tt c} \mapsto {\tt b}
	\end{array}
	\right.
	$$
	on the ternary alphabet $\{ {\tt a,b,c} \}$.
	This morphism is tame but not Arnoux-Rauzy.
	It also belongs to Class $P$ and therefore sends a language closed under reversal to a language closed under reversal.

\medskip	
	However not all tame morphisms preserve richness.
	An interesting question would be to characterize tame morphisms that map rich words to rich words.
	
\medskip	
	Baranwal and Shallit in~\cite{Bar} and~\cite{BaSh19} looked for infinite rich words over $\A$ with the minimal critical exponent.
	They found lower bounds on the minimum value for cardinality of the alphabet $k = 2, 3, 4, 5$.
	In~\cite{CuMoRa20}, Curie, Mol and Rampersad proved that for $k=2$ the suggested bound is the best possible.
	Solving some of the previous questions about richness, may help determining the optimal lower bound for $k \ge 3$.
	
 \subsection*{Acknowledgement}
	The research received funding from the Ministry of Education, Youth and Sports of the Czech Republic through the projects CZ.02.1.01/0.0/0.0/16\_019/0000765 and CZ.02.1.01/0.0/0.0/16\_019/0000778.

\bibliographystyle{fundam}

\begin{thebibliography}{10}
\providecommand{\url}[1]{\texttt{#1}}
\providecommand{\urlprefix}{URL }
\expandafter\ifx\csname urlstyle\endcsname\relax
  \providecommand{\doi}[1]{doi:\discretionary{}{}{}#1}\else
  \providecommand{\doi}{doi:\discretionary{}{}{}\begingroup
  \urlstyle{rm}\Url}\fi
\providecommand{\eprint}[2][]{\url{#2}}
		
		\bibitem{ArRo91}
		 Arnoux P,  Rauzy G.
		\emph{Repr\'esentation g\'eom\'etrique de suites de complexit\'e 2n + 1},
		Bulletin de la Société Mathématique de France,
		1991. 119(2):199--215. doi.org/10.24033/bsmf.2164.
		
		\bibitem{BaPeSt09}
		\v{L}. Balkov\'a, E. Pelantov\'a, \v S. Starosta,
		\emph{Palindromes in infinite ternary words},
		RAIRO - Theoretical Informatics and Applications,  2009. 43(4):687--702.
      doi:10.1051/ita/2009016.
		
		\bibitem{BaPeSt10}
		 Balkov\'a  \v{L},  Pelantov\'a E, Starosta  \v S.
		\emph{Sturmian Jungle (or Garden?) on Multiliteral Alphabets},
		RAIRO - Theoretical Informatics and Applications, 2010. 	44(4):443--470.
     doi:10.1051/ita/2011002.
		
		\bibitem{BaPeSt11}
			 Balkov\'a  \v{L},  Pelantov\'a E, Starosta  \v S.
		\emph{Infinite Words with Finite Defect},
		Advances in Applied Mathematics,  2011. 	47(3):562--574.
      doi:10.1016/j.aam.2010.11.006.
		
		\bibitem{Bar}
		 Baranwal AR.
		\emph{Decision Algorithms for Ostrowski-Automatic Sequences},
		Master Thesis, 	University of Waterloo (2020).   ID: 229124363.
		
		\bibitem{BaSh19}
		 Baranwal AR,  Shallit J.
		\emph{Repetitions in infinite palindrome-rich words},
		in: R. Mercas and D. Reidenbach (eds.),
		Proceedings WORDS 2019, Lecture Notes in Computer Science,
		vol. 11682, Springer, 2019 pp. 93--105.  doi:10.1007/978-3-030-28796-2\_7.
		
		\bibitem{acyclic}
		 Berth{\'e} V,  De Felice C,  Dolce F,  Leroy J,  Perrin D,  Reutenauer C,  Rindone G.
		\emph{Acyclic, connected and tree sets},
		Monatshefte für Mathematik,	2015. 176:521--550.  doi:10.1007/s00605-014-0721-4.
		
		\bibitem{maximalbifixdecoding}
         Berth{\'e} V,  De Felice C,  Dolce F,  Leroy J,  Perrin D,  Reutenauer C,  Rindone G.
		\emph{Maximal bifix decoding},	Discrete Mathematics, 2015.	338(5):725--742.
          doi:10.1016/j.disc.2014.12.010.
		
		\bibitem{specular}
       Berth{\'e} V,  De Felice C, Delecroix V,  Dolce F,  Leroy J,  Perrin D,  Reutenauer C,  Rindone G.
 		\emph{Specular sets},
		Theoretical Computer Science,  2017. 684:3--28.  doi:10.1016/j.tcs.2017.03.001.
		
		\bibitem{BlBrLaVu110}
		 Blondin Mass\'e A,  Brlek S,  Labb\'e S,  Vuillon L.
		\emph{Palindromic complexity of codings of rotations},
		Theoretical Computer Science, 2011.	412(46):6455--6463.
       doi:10.1016/j.tcs.2011.08.007.
		
		\bibitem{BrHaNiRe04}
		 Brlek S,  Hamel S,  Nivat M,  Reutenauer C.
		\emph{On the palindromic complexity of infinite words},
		in: J. Berstel, J. Karhumäki,D. Perrin (Eds.),
		Combinatorics on Words with Applications,
		International Journal of Foundations of Computer Science, 2004. 15(2):293--306.
         doi:10.1142/S012905410400242X.
		
		\bibitem{BuDLuGlZa09}
		 Bucci M,  De Luca A,  Glen A, Zamboni LQ.
		\emph{A connection between palindromic and factor complexity using return words},
		Advances in Applied Mathematics. 2009. 	42(1):60--74.
        doi:10.1016/j.aam.2008.03.005.
		
		\bibitem {Ca94}
		 Cassaigne J.
		\emph{An algorithm to test if a given circular HD0L-language avoids a pattern},
		in: IFIP World Computer Congress’94,  Elsevier, 1994. pp. 	459--464. ID:17440661.
		
		\bibitem {CuMoRa20}
		 Currie JD,  Mol L,  Rampersad N.
		\emph{The repetition threshold for binary rich words},
		Discrete Mathematics \& Theoretical Computer Science, 2020. vol. 22 no 1.
         doi:10.23638/DMTCS-22-1-6.
		
		\bibitem{DrJuPi01}
		 Droubay X,  Justin J,  Pirillo G.
		\emph{Episturmian words and some constructions of de Luca and Rauzy},
		Theoretical Computer Science, 2001. 255(1-2):539--553.
         doi:10.1016/S0304-3975(99)00320-5.
		
		\bibitem{Du2}
		 Durand F.
		\emph{A generalization of Cobham Theorem},
		Theory of Computing Systems, 1998. 31:169--185. Id: hal-00303322.

		
		\bibitem{Du98}
		 Durand F.
		\emph{A characterization of substitutive sequences using return words},
		Discrete Mathematics, 1998. 	179(1-3):89--101. doi:10.1016/S0012-365X(97)00029-0.
		
		\bibitem{Fr99}
		 Frid A.
		\emph{Applying a uniform marked morphism to a word},
		Discrete Mathematics and Theoretical Computer Science,
		1999. 3(3):125--139.  doi:10.46298/dmtcs.255.
		
		\bibitem{FrPuZa13}
		 Frid A,  Puzynina S, Zamboni LQ.
		\emph{On palindromic factorization of words},
		Advances in Applied Mathematics, 2013.	50(5)	737--748. doi:10.1016/j.aam.2013.01.002.

		\bibitem{GlJu09}
		 Glen A,  Justin J.
		\emph{Episturmian words: a survey},
		RAIRO - Theoretical Informatics and Applications,
		2009. 43(3)403--442.  doi:10.1051/ita/2009003.
		
		\bibitem{GlJuWiZa09}
		 Glen A,  Justin J, Widmer S, Zamboni LQ.
		\emph{Palindromic richness}, European Journal of Combinatorics,
		2009. 30(2):510--531.  doi:10.1016/j.ejc.2008.04.006.
		
		\bibitem{GuShSh16}
		 Guo C,  Shallit J, Shur AM.
		\emph{Palindromic rich words and run-length encodings},
		Information Processing Letters, 2016. 116(12):735--738.
       doi:10.1016/j.ipl.2016.07.001.
		
		\bibitem{HaVeZa16}
		 Harju T,  Vesti J,  Zamboni LQ.
		\emph{On a question of Hof, Knill and Simon on palindromic substitutive systems},
		Monatshefte für Mathematik, 2016.	179:379--388. Id: hal-01829318.
		
		\bibitem{HoKnSi95}
		 Hof A,  Knill O,  Simon B.
		\emph{Singular continuous spectrum for palindromic Schrödinger operators},
		Communications in mathematical physics, 1995. 174(1):149--159.
       doi:10.1007/BF02099468.
		
		\bibitem{JuPi02}
		 Justin J,  Pirillo G.
		\emph{Episturmian words and episturmian morphisms},
		Theoretical Computer Science, 2002. 276(1-2):281--313.
     doi:10.1016/S0304-3975(01)00207-9.
		
		\bibitem{Kl12}
		 Klouda K.
		\emph{Bispecial factors in circular non-pushy D0L languages},
		Theoretical Computer Science, 2012. 445(3):63--74.
      doi:10.1016/j.tcs.2012.05.007.
		
		\bibitem{LaPe16}
		 Labb\'e S,  Pelantov\'a E.
		\emph{Palindromic sequences generated from marked morphisms},
		European Journal of Combinatorics, 2016.	51:200--214.
        doi:10.1016/j.ejc.2015.05.006.

		\bibitem{Meyer}
          Meyer CD.
          Matrix Analysis and Applied Linear Algebra, SIAM, 2000.
          ISBN-13:978-0898714548, 10:0898714540.

		\bibitem{Lothaire}
		 Lothaire M.
		\emph{Algebraic Combinatorics on Words}, Cambridge University Press,
		2002. ISBN-13:978-0521180719, 10:0521180716.
		
		\bibitem{MePeVu19}
		 Medkov\'a K,  Pelantov\'a E,  Vuillon L.
		\emph{Derived sequences of complementary symmetric Rote sequences},
		RAIRO - Theoretical Informatics and Applications,
		2019. 53(3-4):125--151. doi:10.1051/ita/2019004.

		
		\bibitem{PeSt17}
		Pelantov\'a E,  Starosta \v S.
		\emph{On Words with the Zero Palindromic Defect},
		In: Brlek S., Dolce F., Reutenauer C., Vandomme É. (eds),
		Combinatorics on Words. WORDS 2017,
		Lecture Notes in Computer Science, 	vol. 10432. Springer, Cham, 2017.
		doi:10.1007/978-3-319-66396-8\_7.

		\bibitem{PytheasFogg}
		 Pytheas Fogg N.
		\emph{Substitutions in dynamics, arithmetics and combinatorics},
		volume 1794 of Lecture Notes in Mathematics. Springer-Verlag (2002).
		Edited by V. Berth{\'e}, S. Ferenczi, C. Mauduit and A. Siegel.
		ISBN-13:978-3540441410, 10:3540441417.

		\bibitem{RuSh}
		 Rubinchik M,  Shur AM.
		\emph{EERTREE: An efficient data structure for processing palindromes in strings},
		European J. Combin., 2018. 	68:249--265.  doi:10.1016/j.ejc.2017.07.021.

		
		\bibitem{Ru17}
		Rukavicka J.
		\emph{On the number of Rich Words},
		In E. Charlier, J. Leroy, M. Rigo (Eds),
		DLT 2017, Lecture Notes in Computer Science,
		vol 10396, Springer (2017),
		345--352.
		
		\bibitem{St16}
		 Starosta \v{S}.
		\emph{Morphic images of episturmian words having finite palindromic defect},
		European Journal of Combinatorics, 2016. 51:359--371.
		
		\bibitem{Ve14}
		 Vesti J.
		\emph{Extensions of rich words}, Theoretical Computer Science,
		2014. 548(4):	14--24. doi:10.1016/ j.tcs.2014.06.033.

	\end{thebibliography}

\end{document}